\documentclass[11pt,letterpaper]{article}
\usepackage[margin=1in]{geometry}
\usepackage[numbers,sort&compress]{natbib}
\usepackage[small]{caption}
\usepackage{graphicx}
\usepackage{amsmath}
\usepackage{amsthm}
\usepackage{enumitem}
\usepackage{framed}
\usepackage{mathrsfs}
\usepackage{microtype}
\usepackage{doi}
\usepackage{hyperref}
\usepackage{thmtools}

\newtheorem{theorem}{Theorem}
\newtheorem{lemma}[theorem]{Lemma}
\newtheorem{property}[theorem]{Property}

\theoremstyle{definition}
\newtheorem{definition}{Definition}

\newcommand{\mysf}[1]{\ensuremath{\mathsf{#1}}}
\newcommand{\ID}{\operatorname{ID}}
\newcommand{\LOCAL}{\mathsf{LOCAL}}
\newcommand{\Pshort}{\mathscr{P}_{\operatorname{short}}}
\newcommand{\Plong}{\mathscr{P}_{\operatorname{long}}}
\newcommand{\midd}{\mathsf{Mid}}
\newcommand{\LabelIn}{\Sigma_{\operatorname{in}}}
\newcommand{\LabelOut}{\Sigma_{\operatorname{out}}}
\newcommand{\Lpump}{ {\ell_{\operatorname{pump}}} }
\newcommand{\Lcount}{\ell_{\operatorname{count}}}
\newcommand{\Lpattern}{\ell_{\operatorname{pattern}}}
\newcommand{\Lwidth}{\ell_{\operatorname{width}}}
\newcommand{\ltape}{\mysf{tape}}
\newcommand{\tape}{\mysf{Tape}}
\newcommand{\lstart}{\mysf{Start}}
\newcommand{\lseparator}{\mysf{Separator}}
\newcommand{\lempty}{\mysf{Empty}}
\newcommand{\lstep}{\mysf{step}}
\newcommand{\lstate}{\mysf{state}}
\newcommand{\lhead}{\mysf{head}}
\newcommand{\linput}{\mysf{Input}}
\newcommand{\loutput}{\mysf{Output}}
\newcommand{\lerror}{\mysf{Error}}
\newcommand{\ldist}{\mysf{dist}}
\newcommand{\successor}{\mysf{succ}}
\newcommand{\predecessor}{\mysf{pred}}
\newcommand{\lin}{\mysf{in}}
\newcommand{\lout}{\mysf{out}}
\newcommand{\currstate}{\mysf{current\_state}}
\newcommand{\transtate}{\mysf{transition\_state}}
\newcommand{\tapecontent}{\mysf{tape\_content}}
\newcommand{\newcontent}{\mysf{new\_content}}
\newcommand{\type}{\mysf{Type}}
\newcommand{\replace}{\mysf{Replace}}
\newcommand{\lba}{\mysf{LBA}}
\newcommand{\lcl}{\mysf{LCL}}
\newcommand{\Enc}{\mathsf{Enc}}
\newcommand{\Dec}{\mathsf{Dec}}

\newcommand{\GG}{\mathcal{G}}

\newcommand{\PP}{\mathcal{P}}

\newcommand{\LL}{\mathcal{L}}

\newcommand{\Pset}{\mathscr{P}}

\newcommand{\simm}{\overset{\star}{\sim}}

\newenvironment{myabstract}
{\list{}{\listparindent 1.5em%
		\itemindent    \listparindent
		\leftmargin    1cm
		\rightmargin   1cm
		\parsep        0pt}%
	\item\relax}
{\endlist}

\newenvironment{mycover}
{\list{}{\listparindent 0pt
		\itemindent    \listparindent
		\leftmargin    1cm
		\rightmargin   1cm
		\parsep        0pt}%
	\raggedright
	\item\relax}
{\endlist}

\newcommand{\myemail}[1]{\,$\cdot$\, {\small #1}}
\newcommand{\myaff}[1]{\,$\cdot$\, {\small #1}\par\medskip}

\hypersetup{
	colorlinks=true,
	linkcolor=black,
	citecolor=black,
	filecolor=black,
	urlcolor=[rgb]{0,0.1,0.5},
	pdftitle={The distributed complexity of locally checkable problems on paths is decidable},
	pdfauthor={Alkida Balliu, Sebastian Brandt, Yi-Jun Chang, Dennis Olivetti, Mika\"el Rabie, Jukka Suomela}
}

\begin{document}
	
	\begin{mycover}
		{\huge\bfseries\boldmath The distributed complexity of locally checkable problems on paths is decidable \par}
		\bigskip
		\bigskip
		
		\textbf{Alkida Balliu}
		\myemail{alkida.balliu@aalto.fi}
		\myaff{Aalto University}
		
		\textbf{Sebastian Brandt}
		\myemail{brandts@ethz.ch}
		\myaff{ETH Zurich}
		
		\textbf{Yi-Jun Chang}
		\myemail{cyijun@umich.edu}
		\myaff{University of Michigan}
		
		\textbf{Dennis Olivetti}
		\myemail{dennis.olivetti@aalto.fi}
		\myaff{Aalto University}
		
		\textbf{Mika\"el Rabie}
		\myemail{mikael.rabie@irif.fr}
		\myaff{Aalto University and IRIF, University Paris Diderot}
		
		\textbf{Jukka Suomela}
		\myemail{jukka.suomela@aalto.fi}
		\myaff{Aalto University}

	\end{mycover}
	
	\medskip
	\begin{myabstract}
		\noindent\textbf{Abstract.}
		Consider a computer network that consists of a path with $n$ nodes. The nodes are labeled with inputs from a constant-sized set, and the task is to find output labels from a constant-sized set subject to some local constraints---more formally, we have an LCL (locally checkable labeling) problem. How many communication rounds are needed (in the standard LOCAL model of computing) to solve this problem?

		It is well known that the answer is always either $O(1)$ rounds, or $\Theta(\log^* n)$ rounds, or $\Theta(n)$ rounds. In this work we show that this question is \emph{decidable} (albeit PSPACE-hard): we present an algorithm that, given any LCL problem defined on a path, outputs the distributed computational complexity of this problem and the corresponding asymptotically optimal algorithm.
	\end{myabstract}

	\thispagestyle{empty}
	\setcounter{page}{0}
	\newpage

	\section{Introduction}\label{sec:introduction}

To what extent is it possible to \emph{automate} the design of algorithms and the study of computational complexity? While algorithm synthesis problems are typically undecidable, there are areas of theoretical computer science in which we can make use of computational techniques in algorithm design---at least in principle, and sometimes also in practice. One such area is the theory of \emph{distributed computing}; see \cite{Dolev2016,Hirvonen2017,Rybicki2015,chang17hierarchy,Brandt2017,Bloem2016,Faghih2015,Klinkhamer2016} for examples of recent success stories. In this work we bring yet another piece of good news:
\begin{framed}
\noindent Consider this setting: there is a computer network that consists of a path with $n$ nodes, the nodes are labeled with inputs from a constant-sized set, and the task is to find output labels from a constant-sized set subject to some local constraints. We show that for any given set of local constraints, \emph{it is decidable} to tell what is the asymptotically optimal number of communication rounds needed to solve this problem (as a function of $n$, for the worst-case input).
\end{framed}

\paragraph{\boldmath Background: $\lcl$s and the $\LOCAL$ Model.}

We focus on what are known as $\lcl$ (\emph{locally checkable labeling}) problems \cite{Naor1995} in the $\LOCAL$ model of distributed computing \cite{Linial1992,Peleg2000}. We define the setting formally in Section~\ref{sec:local}, but in essence we look at the following question:
\begin{itemize}
	\item We are given an unknown input graph of maximum degree $\Delta = O(1)$; the nodes are labeled with \emph{input labels} from a constant-size set $\LabelIn$, and the nodes also have unique identifiers from a polynomially-sized set.
	\item The task is to label the nodes with \emph{output labels} from a constant-size set $\LabelOut$, subject to some \emph{local} constraints $\mathcal{P}$; a labeling is globally feasible if it is locally feasible in all radius-$r$ neighborhoods for some $r = O(1)$.
	\item Each node has to produce its own output label based on the information that it sees in its own radius-$T(n)$ neighborhoods for some function $T$.
\end{itemize}
Here the local constraints $\mathcal{P}$ define an $\lcl$ problem. The rule that the nodes apply to determine their output labels is called a \emph{distributed algorithm} in the $\LOCAL$ model, and function $T(n)$ is the \emph{running time} of the algorithm---here $T(n)$ determines how \emph{far} a node has to see in order to choose its own part of the solution, or equivalently, how many \emph{communication rounds} are needed for each node to gather the relevant information if we view the input graph as a communication network.

In this setting, the case of $T(n) = \Theta(n)$ is trivial, as all nodes can see the entire input. The key question is to determine which problems $\mathcal{P}$ can be solved in sublinear time---here are some examples:
\begin{itemize}
	\item Vertex coloring with $\Delta+1$ colors: can be solved in time $O(\log^* n)$ \cite{Goldberg1988,cole86deterministic} and this is tight \cite{Linial1992,Naor1991}.
	\item Vertex coloring with $\Delta$ colors, for $\Delta > 2$: can be solved in polylogarithmic time \cite{panconesi95delta} and requires at least logarithmic time \cite{chang16exponential} for deterministic algorithms.
\end{itemize}

While the study of this setting was initiated already in the seminal work by Naor and Stockmeyer in 1995 \cite{Naor1995}, our understanding of these questions has rapidly advanced in the past three years \cite{Balliu2018stoc,Balliu2018disc,Brandt2016,chang16exponential,chang17hierarchy,fischer17sublogarithmic,ghaffari17distributed,Ghaffari2018,Ghaffari2018a,Pettie2018}. The big surprises have been these:
\begin{itemize}
	\item There are $\lcl$ problems with infinitely many different time complexities---for example, we can construct $\lcl$ problems with a time complexity exactly $\Theta(n^{\alpha})$ for any rational number $0 < \alpha \le 1$.
	\item Nevertheless, there are also wide gaps in the complexity landscape: for example, no $\lcl$ problem has a (deterministic) computational complexity that is between $\omega(\log^* n)$ and $o(\log n)$.
\end{itemize}

However, what is perhaps most relevant for us is the following observation: if we look at the case of $\Delta = 2$ (paths and cycles), then the time complexity of any $\lcl$ problem is either $O(1)$, $\Theta(\log^* n)$, or $\Theta(n)$, and the same holds for both deterministic and randomized algorithms \cite{Naor1995,Brandt2017,chang17hierarchy}.

\paragraph{\boldmath Decidability of $\lcl$ Time Complexities.}

For a fixed $\Delta$, any $\lcl$ problem has a trivial finite representation: simply enumerate all feasible radius-$r$ local neighborhoods. Hence it makes sense to ask whether, given an $\lcl$ problem, it is possible to determine its time complexity. The following results are known by prior work:
\begin{itemize}
	\item If the input graph is an \emph{unlabeled path or cycle}, the time complexity is decidable \cite{Naor1995,Brandt2017}.
	\item If the input graph is a \emph{grid or toroidal grid}, the time complexity is undecidable \cite{Naor1995}. However, there are also some good news: in unlabeled toroidal grids, the time complexity falls in one of the classes $O(1)$, $\Theta(\log^* n)$, or $\Theta(n)$, it is trivial to tell if the time complexity is $O(1)$, and it is semi-decidable to tell if it is $\Theta(\log^* n)$ \cite{Brandt2017}.
	\item In the case of trees, there are infinitely many different time complexities, but there is a gap between $\omega(\log n)$ and $n^{o(1)}$, and it is decidable to tell on which side of the gap a given problem lies \cite{chang17hierarchy}.
\end{itemize}
Somewhat surprisingly, the seemingly simple case of \emph{labeled paths or cycles} has remained open all the way since the 1995 paper by Naor and Stockmeyer \cite{Naor1995}, which defined $\lcl$s with inputs but analyzed decidability questions only in the case of unlabeled graphs.

We initially expected that the question of paths with input labels is a mere technicality and the interesting open questions are related to much broader graph families, such as rooted trees, trees, and bounded-treewidth graphs. However, it turned out that the \emph{main obstacle for understanding decidability in any such graph family seems to lie in the fact that the structure of the graph can be used to encode arbitrary input labels}, hence it is necessary to first understand how the input labels influence decidability---and it turns out that this makes all the difference in the case of paths.

In this work we show that the time complexity of a given $\lcl$ problem on \emph{labeled paths or cycles} is decidable. However, we also show that decidability is far from trivial: the problem is PSPACE-hard, as $\lcl$ problems on labeled paths are expressive enough to capture linear bounded automata (Turing machines with bounded tapes).

	\section{Model}\label{sec:local}
	
	\paragraph{\boldmath The $\LOCAL$ Model.}
	The model of computation we consider in this work is the $\LOCAL$ model of distributed computing \cite{Linial1992,Peleg2000}.
	In the $\LOCAL$ model, each node of the input graph is considered as a computational entity that can communicate with the neighboring nodes in order to solve some given graph problem.
	Computation is divided into synchronous rounds, where in each round each node first sends messages of arbitrary size to its neighbors, then receives the messages sent by its neighbors, and finally performs some local computation of arbitrary complexity.
	Each node is equipped with a globally unique identifier ($\ID$) which is simply a bit string of length $O(\log n)$, where $n$ denotes the number of nodes of the input graph.
	In the beginning of the computation, each node is aware of its own $\ID$, the number of nodes and the maximum degree $\Delta$ of the input graph, and potentially some additional problem-specific input.
	Each node has to decide at some point that it terminates, upon which it returns a local output and does not take part in any further computation; the problem is solved correctly if the local outputs of all nodes together constitute a global output that satisfies the output constraints of the given problem.
	
	Each node executes the same algorithm; the running time of the distributed algorithm is the number of rounds until the last node terminates.
	It is well known that, due to the unbounded message sizes, an algorithm with runtime $T(n)$ can be equivalently described as a function from the set of all possible radius-$T(n)$ neighborhoods to the set of allowed outputs.
	In other words, we can assume that in a $T(n)$-round algorithm, each node first gathers the topology of and the input labels contained in its radius-$T(n)$ neighborhood, and then decides on its output based solely on the collected information.
	
	\paragraph{Locally Checkable Labelings.}
	The class of problems we consider is \emph{locally checkable labeling} ($\lcl$) problems \cite{Naor1995}.
	$\lcl$ problems are defined on graphs of bounded degree, i.e., we will assume that $\Delta = O(1)$.
	Formally, an $\lcl$ problem is given by a finite input label set $\LabelIn$, a finite output label set $\LabelOut$, an integer $r$, and a finite set $\mathcal C$ of graphs where every node is labeled with a pair $(\ell_{\operatorname{in}}, \ell_{\operatorname{out}}) \in \LabelIn \times \LabelOut$ and one node is marked (as the center).
	Each node of the input graph is assigned an input label from $\LabelIn$ before the computation begins, and the global output of a distributed algorithm is correct if the radius-$r$ neighborhood of each node $v$, including the input labels given to the contained nodes and the output labels returned by the contained nodes, is isomorphic to an element of $\mathcal C$ where $v$ corresponds to the node marked as the center.

	In the case of directed paths as our class of input graphs, we are interested in identifying the simplest possible form of $\lcl$ problems. For this purpose, we define \emph{$\beta$-normalized} $\lcl$s; these are problems for which the input is just binary, and the size of the set of output labels is $\beta$. Moreover, the solution can be checked at each node $v$ by just inspecting the input and output of $v$, and, separately, the output of $v$ and the output of its predecessor. More formally, a \emph{$\beta$-normalized} $\lcl$ problem is given by finite input and output label sets $\LabelIn$, $\LabelOut$ satisfying $|\LabelIn| = 2$, $|\LabelOut| = \beta$, a finite set $\mathcal C_{\operatorname{in}-\operatorname{out}}$ of pairs $(\ell_{\operatorname{in}}, \ell_{\operatorname{out}}) \in \LabelIn \times \LabelOut$ and a finite set $\mathcal C_{\operatorname{out}-\operatorname{out}}$ of pairs $(\ell_{\operatorname{out}}, \ell'_{\operatorname{out}}) \in \LabelOut \times \LabelOut$.
	The global output of a distributed algorithm for the $\beta$-normalized $\lcl$ problem is correct if the following hold:
	\begin{itemize}
		\item For each node $v$, we have $(\linput(v), \loutput(v)) \in \mathcal C_{\operatorname{in}-\operatorname{out}}$, where $\linput(v)$ denotes the input label of $v$, and $\loutput(v)$ the output label of $v$.
		\item For each node $v$ that has a predecessor, we have $(\loutput(v), \loutput(u)) \in \mathcal C_{\operatorname{out}-\operatorname{out}}$, where $u$ is the predecessor of $v$, and $\loutput(v), \loutput(u)$ are the output labels of $v$ and $u$,  respectively.
	\end{itemize}
	It is straightforward to check that a $\beta$-normalized $\lcl$ problem is indeed a special case of an $\lcl$ problem where $r = 1$.

	\section{Hardness}\label{sec:lb}
	In this section we study the hardness of determining the distributed complexity of \lcl{}s on paths and cycles with input labels. More precisely, we start by proving the existence of a family $\Pi$ of \lcl{} problems for consistently globally oriented paths, such that, given an \lcl{} problem in $\Pi$, it is PSPACE-hard to decide if its distributed complexity is $O(1)$ or $\Theta(n)$. Our main result shows the following.
	\begin{framed}
		\noindent It is PSPACE-hard to distinguish whether a given \lcl{} problem $\PP$ with input labels can be solved in $O(1)$ time or needs $\Omega(n)$ time on globally oriented path graphs.	
	\end{framed}
	
	The high level idea of the proof of the above result is as follows. We would like to encode the execution of Turing machines as \lcl{}s on consistently oriented paths, and then define some $\lcl{}$ for which the complexity depends on the running time of the machine. This is fairly easy on oriented grids, for example, where we can use one dimension of the grid as a tape, and the other dimension as time. One may try to do the same on paths, by projecting everything on a single dimension, concatenating the tape state of each step. Unfortunately, the obtained encoding is not locally checkable, since the length of the tape may be non-constant. Hence, in order to guarantee the local checkability, we should consider Turing machines having a tape of size at most $B$, where $B$ is a constant with respect to the number of nodes in the path where we want to encode its execution. For this purpose, we consider Linear Bounded Automata (\lba)~\cite[p.~225]{HU79}. An \lba{} is a Turing machine that has a tape of size upper bounded by some $B$. We show that, if $B$ is constant with respect to the number of nodes in the path, we can then encode the execution of an \lba{} $M_B$ as an \lcl{} for directed paths. Moreover, we show that by seeing this encoding as a two party game between a \emph{prover} and a \emph{disprover}, we can encode the execution of $M_B$ using labels of constant size that do not depend on $B$, even in the case in which the \lcl{} checkability radius is $1$. If the execution of  $M_B$ is not correctly encoded in the input of the \lcl, then we can disprove its correctness using output labels of size $O(B)$. Moreover, we ensure that, if the execution of $M_B$ is correctly encoded in the input of the \lcl, it is not possible to produce a correct proof of non-correctness. Then, in order to obtain an \lcl{} with a distributed complexity that depends on the execution time of $M_B$, we encode some secret input at the first node of the path. We require then that all nodes involved in a correct encoding must produce the same secret as output.
	
	\begin{figure}
		\centering
		\includegraphics[width=\textwidth]{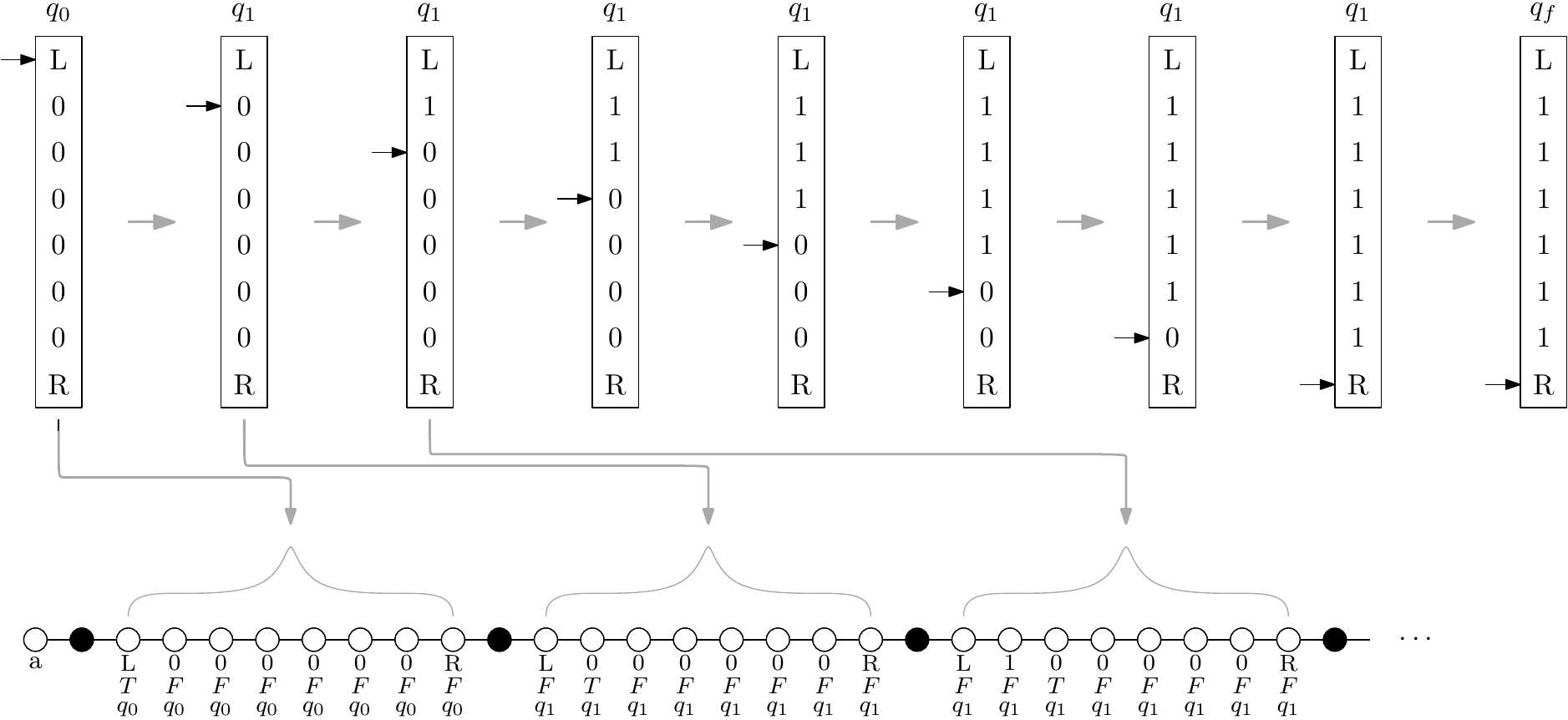}
		\caption{Illustration of a correct encoding of the execution of an \lba{} on a path; black nodes act as separators between the encoding of two consecutive steps of the \lba; in the example, the \lba{} executes a unary counter.}\label{good_input}
	\end{figure}

	Figure \ref{good_input} shows an example of an \lba{} that executes a unary counter, and its encoding as input to nodes on a path. In this instance, all nodes must produce the symbol $a$ as output. Figure  \ref{error_example}  shows an example of the wrong input (the tape has been copied incorrectly between two consecutive steps of the \lba{}). In this case, nodes are allowed to produce a chain of errors. Different types of errors will be handled using different types of error labels. In the example, all nodes that produce the error chain, output $E^2$, indicating an error of type $2$. We will show that we need $O(B)$ symbols to handle all possible errors (including the case in which the input tape is too long, way more than~$B$). Also, it is necessary that all error chains that we allow as outputs must be locally checkable.
	
	\begin{figure}
		\centering
		\includegraphics[width=\textwidth]{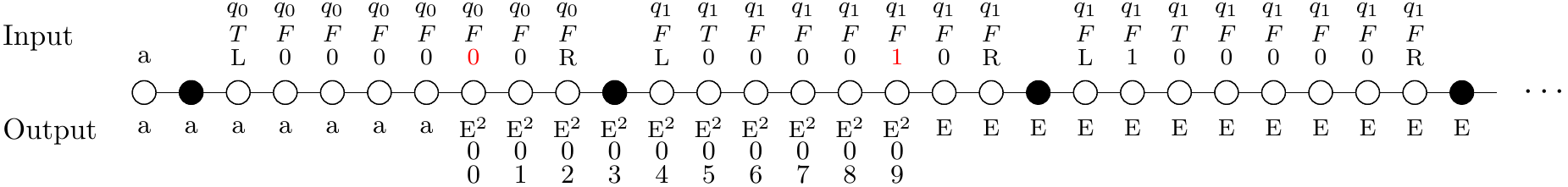}
		\caption{Illustration of an incorrect encoding of the execution of an \lba{} on a path; in the example, the tape of the \lba{} is wrongly copied (the inputs in red are different, while they should be the same). The error output $E^2$ encodes the distance of $B+1$ between the two nodes, and the input wrongly copied.}\label{error_example}
	\end{figure}
	
	Another interesting problem is to identify, for an \lcl{} that can be distributedly solved in constant time, how big this constant can be. In particular, we first focus on identifying the simplest possible description of an \lcl{}, and then, we provide a lower bound on the complexity of a constant time \lcl{}, as a function of the size of the \lcl{} description. For this purpose, we consider $\beta$-normalized \lcl{}s, i.e., problems for which the input labeling is just binary and there are $\beta$ possible output labels. Also, the verifier for these \lcl{}s is the simplest possible: it can only check if the output of a node is correct w.r.t.\ its input, and separately, if the output of a node is correct w.r.t.\ the output of its predecessor. Therefore, we show how to convert an \lcl{} to a $\beta$-normalized one by encoding the input in binary (Figure \ref{normalization} shows an example), and obtain the following result.
	\begin{framed}
		\noindent There are $\beta$-normalized \lcl{}s that can be solved in constant time but the distributed time complexity is $2^{\Omega(\beta)}$.
	\end{framed}

	\begin{figure}
		\centering
		\includegraphics[width=0.8\textwidth]{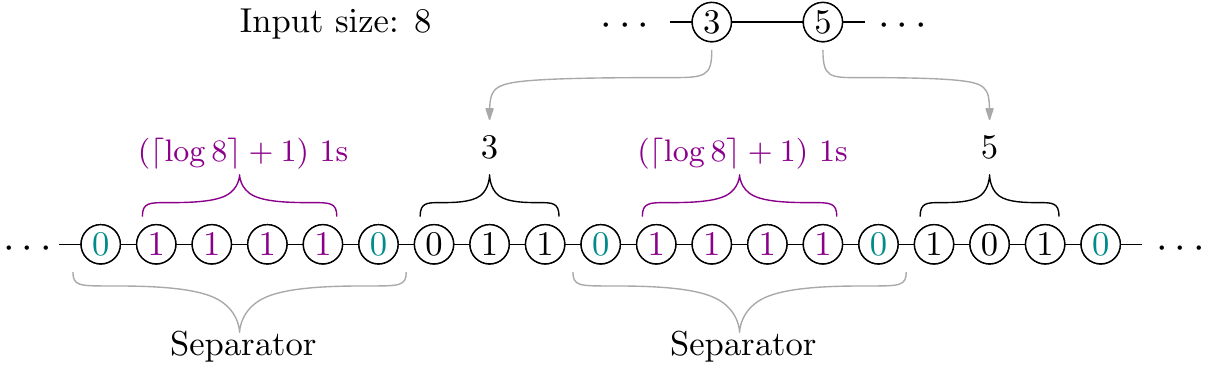}
		\caption{Illustration of the normalization of an \lcl.}\label{normalization}
	\end{figure}
	
	All results that we have been described so far apply to globally oriented paths. Nevertheless, we show that ideas and techniques can be generalized to work on undirected path and cycles as well, obtaining essentially the same results. Finally, we will show how to lift these results to trees \emph{without} input labels, proving the following result.
	\begin{framed}
		\noindent It is PSPACE-hard to distinguish whether a given \lcl{} problem $\PP$ without input labels can be solved in $O(1)$ time or needs $\Omega(n)$ time on trees with degree $\Delta = 3$.
	\end{framed}

\subsection{Linear Bounded Automata}
A Linear Bounded Automata $M_B$ is a Turing Machine having a \emph{bounded tape} of size at most $B$, such that it is able to recognize the boundaries of the tape~\cite[p.~225]{HU79}. More formally, we define an \lba{} as a tuple of $5$ elements $M = (Q,q_0,q_f,\Gamma,\delta)$, where
\begin{itemize}[noitemsep]
	\item $Q$ is a finite set of states;
	\item $q_0 \in Q$ is the initial state;
	\item $q_f \in Q$ is the final state;
	\item $\Gamma$ is a finite set of tape alphabet symbols that contains integers $0$, $1$, and special symbols $L$ (\emph{left}), and $R$ (\emph{right});
	\item $\delta$ is the transition function, where $\delta \colon Q\setminus\{q_f\} \times \Gamma \rightarrow Q \times \Gamma \times \{-,\leftarrow,\rightarrow\}$.
\end{itemize}
The tape of $M_B$ is initialized as follows:
\begin{itemize}[noitemsep]
	\item the first cell is marked with the symbol $L$;
	\item the last cell is marked with the symbol $R$;
	\item all other cells contain an integer in $\{0,1\}$.
\end{itemize}
An execution of an \lba{} is a sequence $(\lstep_i ~|~ i \in \{1,\ldots, t\})$, where
\begin{itemize}[noitemsep]
	\item $\lstep_i = (\lstate_i,\ltape_i,\lhead_i)$;
	\item $\lstate_t = q_f$;
	\item $\delta(\lstate_i,\ltape_i[\lhead_i]) = (\lstate_{i+1},\ltape_{i+1}[\lhead_i],\epsilon)$, and $\lhead_{i+1} $ is
	\begin{itemize}[noitemsep]
		\item $\lhead_{i}-1 $ if $\epsilon$ is $\leftarrow$;
		\item $\lhead_{i}$ if $\epsilon$ is $-$;
		\item $\lhead_{i}+1$ if $\epsilon$ is $\rightarrow$.
	\end{itemize}
\end{itemize}

\subsection{\boldmath The \lcl{} Problem}
We define a family $\Pi$ of \lcl{}s, in which each problem $\Pi_{M_B}$ depends on the \lba{} $M_B$. The general idea is that the input of the \lcl{} may encode the execution of an \lba{} $M_B$. If it is the case, nodes are required to solve a problem that requires a time proportional to the execution time of $M_B$. On the other hand, if it is not the case, nodes can produce an output that proves that this encoding is wrong. In order to define valid \lcl{}s, we consider the case where $B=O(1)$, that is, the size of the tape does not depend on the size of the distributed network.

\subsubsection{Input Labels}
We define the input labels of our \lcl{} as follows:
\begin{itemize}[noitemsep]
	\item $\lstart(\phi)$, where $\phi\in\{a,b\}$, indicates a symbol that will be used as some kind of secret;
	\item \lseparator, a label that acts as a separator between two steps of $M_B$;
	\item $\tape(c, s, h)$ gives information about the tape and the state of $M_B$, where the content $c\in\{0,1,L,R\}$, the state $s\in Q$, and the head $h\in \{\mbox{true}, \mbox{false}\}$;
	\item \lempty{}, indicating an empty input.
\end{itemize}
Note that the size of the set of possible input labels does not depend on the size $B$ of the tape.

\subsubsection{\boldmath Encoding an \lba{} on a Path}
Suppose we have a consistent global orientation in the path $P=(p_0,p_1, p_2, \ldots, p_{n-1})$. Let $\lstep = \big(\lstep_i = (\lstate_i,\ltape_i,\lhead_i) ~|~ i \in \{1,\ldots, t\}\big)$ be the execution of the \lba{} $M_B$ starting from a tape initialized with $(L,0,\ldots,0,R)$.
\begin{definition}
	The input of the \lcl{} is a \textbf{good input} if the first node of the path has in input $\lstart(\phi)$, where $\phi\in\{a,b\}$, and the rest of the path correctly encodes the execution of an \lba{} $M_B$ initialized with $(L,0,\ldots,0,R)$ (see Figure \ref{good_input}). More precisely:
	\begin{itemize}[noitemsep]
		\item $\linput(p_0)  = \lstart(\phi)$;
		
		\item $\linput(p_{(i-1)(B+1)+1}) = \lseparator$ for $i \in {1,\ldots,t}$;
		
		\item $\linput(p_{(i-1)(B+1)+1+j}) = \tape(c,s,h)$ for $i \in {1,\ldots,t}$, $j \in \{1,\ldots,B\}$, where
		\begin{itemize}[noitemsep]
			\item $c = \ltape_i[j]$;
			\item $s = \lstate_i$;
			\item $h = \mbox{true}$ if $\lhead_i = j$, otherwise $h = \mbox{false}$;
		\end{itemize}
		\item All other nodes have in input $\lempty$.
	\end{itemize}
\end{definition}

\subsubsection{Output Labels}
The set of output labels is the following.
\begin{itemize}[noitemsep]
	\item $\lstart(\phi)$;
	
	\item $\lempty$;
	
	\item $\lerror$: a generic error label;
	
	\item $\lerror^0(i)$ where $0\le i \le B+1$: an error of type $0$ indicating that the machine is not correctly initialized;
	
	\item $\lerror^1(i)$, where $0\le i \le B$: an error of type $1$ that we will use in the case where the size of the tape is not correct, i.e., when the size of the tape is not $B$;
	
	\item $\lerror^2(x, i)$, where $x\in\{0,1,L,R\}$ and $0\le i \le B+1$: an error of type $2$ used when the tape of $M_B$ is wrongly copied;
	
	\item $\lerror^3$: an error of type $3$ is used in case nodes have inconsistent states;
	
	\item $\lerror^4(\currstate, \tapecontent, i)$, where $0\le i \le B+2$: an error of type $4$ indicating that the transition of $M_B$ is encoded incorrectly (this error captures also the case where the head is missing);
	
	\item $\lerror^5(x)$ where $x \in \{0,1\}$: an error of type 5 used in the case when there is more than one head.
	
\end{itemize}

\subsubsection{\boldmath \lcl{} Constraints}
The high level idea is the following. If the path encodes a good input, then nodes that are not labeled $\lempty$ are required to output the input given to the first node of the path (either $a$ or $b$). Otherwise, nodes can produce a locally checkable proof of an error (see Figure \ref{error_example} for an example). While nodes may output $a$ or $b$ even in the case in which the input is not a good input, nodes must not be able to produce a proof error in the case in which the path encodes a good input. We describe all these requirements as locally checkable constraints.

An output labeling for problem $\Pi_{M_B}$ is correct if the following conditions are satisfied for nodes of the path $P=(p_0,\ldots,p_{n-1})$. Note that, although nodes do not know their position on the path, for the sake of simplicity we will denote with $p_{i-1}$ the predecessor of $p_i$, if it exists.

\begin{enumerate}
	\item Each node $v$ produces exactly one output label.
	
	\item If $\loutput(v)=\lempty$ then $\linput(v)=\lempty$.
	
	\item If $v$ has no predecessors (i.e., $v=p_0$) and $\loutput(v)=\lstart(\phi)$, then $\linput(v)=\lstart(\phi)$.
	
	\item If $\loutput(p_{i-1}) = a$ then $\loutput(p_{i}) \neq b$, and if $\loutput(p_{i-1}) = b$ then $\loutput(p_{i}) \neq a$.
	
	\item If $\loutput(p_i)=\lerror^0(j)$, then
	\begin{itemize}
		\item if $j=0$ then the node has no predecessor;
		\item if $j>0$ then $\loutput(p_{i-1})= \lerror^0(j-1)$.
	\end{itemize}
	
	\item If $\loutput(p_i)=\lerror^1(j)$, then
	\begin{itemize}
		\item if $j=0$, then $\linput(p_i)=\lseparator$;
		\item if $j> 0$ then  $\linput(p_i)\neq\lseparator$ and  $\loutput(p_{i-1})=\lerror^1(j-1)$.
	\end{itemize}
	
	\item If $\loutput(p_i)=\lerror^2(x, j)$, then
	\begin{itemize}
		\item if $j=0$, then $\linput(p_i) = \tape(c,s,h)$ where $h=\mbox{false}$, $c=x$;
		\item if $j=B+1$ then $\linput(p_i) = \tape(c,s,h)$ where $c \neq x$;
		\item if $0<j< B+1$ then $\loutput(p_{i-1})=\lerror^2(x,j-1)$.
	\end{itemize}
	
	\item If $\loutput(p_i)=\lerror^3$, then $\linput(p_i)=\tape(c,s,h)$, $\linput(p_{i-1})=\tape(c',s',h')$, and $s \neq s'$.	
	
	\item If $\loutput(p_i)=\lerror^4(\currstate, \tapecontent, j)$, let $(\transtate,\newcontent,\epsilon) = \delta(\currstate,\tapecontent)$
	\begin{itemize}
		\item if $j=0$, then $\linput(p_i)=\tape(c,s,h)$ where $c=\tapecontent$, $s=\currstate$, $h=\mbox{true}$;
		\item if $j= B$ and $\epsilon = {\leftarrow}$, or $j= B+1$ and $\epsilon = -$, or $j=B+2$ and $\epsilon = {\rightarrow}$ (i.e., if node $p_i$ is an ``$\lerror^4$ final node''), then either $\currstate$ is a final state or $\linput(p_i)=\tape(c,s,h)$ where $s \neq \transtate$ or $h = \mbox{false}$;
		\item otherwise, then $\loutput(p_{i-1})=\lerror^4(\currstate, \tapecontent, j-1)$.
	\end{itemize}
	
	\item If $\loutput(p_i)=\lerror^5(x)$
	\begin{itemize}
		\item if $\loutput(p_{i-1})\neq \lerror^5$ then $\linput(p_i)=\tape(c,s,h)$ where $h=\mbox{true}$ and $x=0$.
	\end{itemize}
	
	\item If $\loutput(p_i)=\lerror$ then one of the following condition holds:
	\begin{itemize}
		\item $\linput(p_i)\neq\lstart(\phi)$ and $p_i$ has no predecessors;
		\item $\linput(p_i)=\lstart(\phi)$ and $p_i$ has a predecessor;
		\item $\linput(p_{i-1})$ or $\loutput(p_{i-1})$ is $\lempty$;
		\item $\loutput(p_{i-1}) = \lerror$;
		\item $\loutput(p_{i-1}) = \lerror^0(j)$, $j>0$, and
		\begin{itemize}
			\item if $j=1$ then $\linput(p_{i-1}) \neq \lseparator$;
			\item if $j\ge 2$ then either $\linput(p_{i-1})\neq \tape$, or $\linput(p_{i-1})=\tape(c,s,h)$ and:
			\begin{itemize}
				\item if $j=2$ either $c\neq L$, or $s\neq q_0$ or $h=\mbox{false}$;
				\item if $2< j\le B$ either $c\neq 0$, or $s\neq q_0$ or $h=\mbox{true}$;
				\item if $j = B+1$, either $c\neq R$, or $s\neq q_0$ or $h=\mbox{true}$;
			\end{itemize}
		\end{itemize}
		
		\item $\linput(p_i)=\lseparator$ and $\loutput(p_{i-1})= \lerror^1(x)$ where $x \neq B$;
		
		\item $\linput(p_i)\neq\lseparator$ and $\loutput(p_{i-1}) = \lerror^1(B)$;
		
		\item $\loutput(p_{i-1}) = \lerror^2(x,j)$ where $j = B + 1$;

		\item $\loutput(p_{i-1}) = \lerror^3$;
		
		\item $p_{i-1}$ is an ``$\lerror^4$ final node'';
		
		\item $\loutput(p_{i-1}) = \lerror^5(x)$ and $\linput(p_{i-1})=\tape(c,s,h)$ where $h=\mbox{true}$ and $x=1$.
	\end{itemize}
	
	\item If $\loutput(p_{i})$ is of type $\lerror^x$, then $\loutput(p_{i-1})$ must not be of type $\lerror^y$ where $y \neq x$.
\end{enumerate}
The following property directly holds by definition of the constraints.
\begin{property}
	Each node is able to locally check all constraints by just inspecting its own input and output, and the ones of its predecessor (if it exists).
\end{property}

\subsection{\boldmath Upper Bound on the Complexity of the \texorpdfstring{$\lcl{}$}{LCL}}
We need to consider two possible scenarios: either $M_B$ terminates within time $T$, or $M_B$ loops. In the case in which $M_B$ loops, we show a simple $O(n)$ algorithm that solves the \lcl{} $\Pi_{M_B}$. As we know, any problem for which a solution exists can be solved in $O(n)$ rounds in the $\LOCAL$ model by gathering all the graph and solving the problem locally. There always exists a solution for problem $\Pi_{M_B}$ if $M_B$ loops, in fact:
\begin{itemize}[noitemsep]
	\item If $\linput(p_0)=\lstart(\phi)$, then all nodes output $\phi$, even if there are errors in the machine encoding.
	\item Otherwise, if $\linput(p_0)\neq\lstart(\phi)$, all nodes output $\lerror$.
\end{itemize}
It is easy to see that this output satisfies the \lcl{} constraints described above.

Suppose that $M_B$ terminates. In this case, we show how to solve the \lcl{} problem $\Pi_{M_B}$ in constant time. More precisely, if $M_B$ terminates in $T$ rounds, we show a distributed algorithm that solves $\Pi_{M_B}$ in $T' = 2+(B+1)T$ rounds. Each node $v$ starts by gathering its $T'$-radius neighborhood $B_v(T')$. Notice that, by definition, if the input is a good input, then for each node $v$ that is taking part in the encoding of the execution of $M_B$ (i.e., $\linput(v)\neq \lempty$), $B_v(T')$ contains $p_0$. Hence, if a node $v$ does not see $p_0$ after gathering its ball $B_v(T')$, it means that the input is not a good input. So, after gathering its $T'$-radius ball, each node $v$ does the following.
\begin{itemize}[noitemsep]
	\item If $\linput(v)=\lempty$, then $\loutput(v)=\lempty$.
	\item If $B_v(T')$ does not contain $p_0$, or if $\linput(p_0) \neq \lstart(\phi)$, then $v$ outputs $\lerror$.
	\item If $B_v(T')$ is a good input, then $v$ outputs $\lstart(\phi)$.
\end{itemize}
The remaining case that we still need to handle is when $B_v(T')$ contains $p_0$, $\linput(p_0) = \lstart(\phi)$, but $B_v(T')$ does not look like a good input. We want nodes to produce a proof of an error in some consistent way. Thus, we show that nodes can identify the first error and produce a proof based on that. First of all, notice that, since $v$ sees the first node in the path, $v$ can compute its position $i$ on the path. Also, node $v$ can identify who is the first node $u$ not satisfying the constraints of being a good input. Let $j$ be the position of $u$ in the path, that is $u=p_j$. Now we distinguish the following cases based on $B_u(B+2)$ (the output of each node will be determined by the first case encountered in the following list).
\begin{enumerate}
	\item If $\linput(p_j)=\lstart(\phi)$ and $j\neq 0$, then, if $i< j$, $\loutput(v) = \lstart(\phi)$; otherwise $\loutput(v)=\lerror$.
	
	\item If $j \le B+1$, it means that either the initial state is encoded incorrectly, or the tape is not initialized correctly, or the head is not initialized on the correct position. In this case, if $i\le j$, then $\loutput(v) = \lerror^0(i)$, otherwise $\loutput(v)=\lerror$.
	
	\item If $\linput(p_{j-(B+1)}) = \lseparator$ and $\linput(p_j) \neq \lseparator$, then the length of the tape is too long, and $u$ expected to have in input \lseparator. Then, if $i < j-(B+1)$, $\loutput(v)=\lstart(\phi)$; if $i > j$ then $\loutput(v)=\lerror$; otherwise, $\loutput(v)=\lerror^1(i - j + B + 1)$.
	
	\item If $\linput(p_{j}) = \lseparator$ and there exists a $k$ such that $1 \le j-k < B+1$ such that $\linput(p_{k}) = \lseparator$, then the length of the tape is too short, and $u$ did not expect to have a separator. In this case, if $i < k$ then $\loutput(v)=\lstart(\phi)$; if $i \ge j$ then $\loutput(v)=\lerror$; otherwise $\loutput(v)=\lerror^1(k-i)$.
	
	\item If $\linput(p_{j-(B+1)}) = \tape(c,s,h)$ where $h=\mbox{false}$, $c=x$, and $\linput(p_j) = \tape(c',s',h')$, where $c'\neq x$, then the tape of $M_B$ has been copied incorrectly. In this case, if $i < j-(B+1)$, then $\loutput(v)=\lstart(\phi)$; if $i > j$ then $\loutput(v)=\lerror$; otherwise, $\loutput(v)=\lerror^2(x, i - j + B + 1)$.
	
	\item If $\linput(p_{j}) = \tape(c,s,h)$ and  $\linput(p_{j-1}) = \lseparator$ and there exists a $k < j + B$ such that $\linput(p_{k}) = \tape(c',s',h')$ and that $s \neq s'$, it means that nodes have inconsistent states. Consider the minimum $k$ satisfying the constraints. If $i<k$ then $\loutput(v)=\lstart(\phi)$; if $i > k$ then $\loutput(v)=\lerror$; otherwise, $\loutput(v)=\lerror^3$.
	
	\item If none of the above is satisfied, it means that there exist a $k$ satisfying $j-k \le B+2$, such that $\linput(p_k) = \tape(c,s,h)$ and $h = \mbox{true}$. Let $(\transtate,\newcontent,\epsilon) = \delta(s,c)$. It holds that if $\epsilon$ is $\leftarrow$, $-$, or $\rightarrow$, then $j-k$ is respectively $B$, $B+1$, or $B+2$. If $\linput(p_j) = \tape(c',s',h')$, where either $h' = \mbox{false}$, or $\transtate \neq s'$, or $s$ is a final state, then there is some error in the transition (this captures also the case where there is no head). If $i < k$, then $\loutput(v)=\lstart(\phi)$; if $i > j$ then $\loutput(v)=\lerror$; otherwise, $\loutput(v)=\lerror^4(s, c, k-i)$. Notice that this case captures also the one where the head is missing.
	
	\item If $\linput(p_j)=\tape(c,s,h)$ where $h=\mbox{true}$, since all the above cases are not satisfied, it means that there exists a $k$, such that $\linput(p_k)=\tape(c',s',h')$, $h'=\mbox{true}$, $|j-k| < B$ and all nodes $p_{min(j,k)},\ldots,p_{max(j,k)}$ are labeled with some $\tape$. That is, there are at least two heads, one on node $p_j$ and one on node $p_k$. In this case, if $i < min(j,k)$, then $\loutput(v)=\lstart(\phi)$; if $i > max(j,k)$ then $\loutput(v)=\lerror$; if $i = min(j,k)$ then $\loutput(v)=\lerror^5(0)$, otherwise $\loutput(v)=\lerror^5(1)$.
\end{enumerate}
If the path encodes a good input, every node taking part in the encoding of the execution of $M_B$ outputs $\lstart(\phi)$, and in this case it is easy to see that the output satisfies the \lcl{} constraints.

Therefore, assume that the path does not correctly encode the execution of $M_B$ starting from the correct tape content. First of all, notice that the algorithm handles all possible errors in the machine encoding, that is, if the input is not good, at least one case of the list applies. Consider all nodes $v$ that do not have in input $\lempty$, that is, all nodes taking part in the encoding of the execution of $M_B$. If node $v$ sees the first node $p_0$, i.e., if the distance between $p_0$ and $v$ is at most $T'$ (notice that a good input has length $T'$), then it is easy to see that the output satisfies the \lcl{} constraints. Some care is needed in the case where a node $v$ outputs a generic error $\lerror$ and $v$ does not see $p_0$: we need to show that also in this case the output is valid, meaning that the \lcl{} constraints are satisfied. In this case, the distance between $p_0$ and $v$ is strictly greater than $T'$, and since the encoding of the execution of $M_B$ is not correct, then
\begin{itemize}[noitemsep]
	\item either the path $(p_0,\ldots,v)$ does not correctly encode the execution of $M_B$,
	\item or $M_B$ is not correctly initialized and it loops.
\end{itemize}
In the first case, some node on the path between $p_0$ and $v$ will output some specific error $\lerror^x$ where $x>0$, while in the second case initial nodes will output $\lerror^0$. In both scenarios the constraints for $\lerror$ are satisfied. The complexity of the algorithm is $O(B \cdot T)$.

\subsection{\boldmath Lower Bound on the Complexity of the \texorpdfstring{$\lcl{}$}{LCL}}
Let us define $T''$ as follows. If $M_B$ terminates in time $T$, then $T'' =T' = 2+(B+1)T$. If $M_B$ loops, then $T'' = n$. We prove a lower bound on the complexity of $\Pi_{M_B}$ of $\Omega(T'')$ rounds, by showing that $\Omega(T'')$ rounds are needed in the case where the input is a good input. In particular, we show that, in a good input, for all nodes $v$ such that $\linput(v)\neq \lempty$, $\loutput(v)$ must be $\lstart(\phi)$. The result then comes from the fact that, for some nodes, it requires $\Omega(T'')$ rounds in order to see if $\phi=a$ or $\phi=b$.

First of all, we ignore nodes that have in input $\lempty$ since, in a good input, they are at distance at least $T''+1$ from $p_0$, the first node of the path. Hence, assume that a node $v$ not having $\lempty$ in input does not output $\lstart(\phi)$. In this case, $v$ can either output a generic error $\lerror$, or a specific error $\lerror^x$. If all nodes output $\lerror$, the verifier rejects on $p_0$. If all nodes, starting from a node $p_i$ where $i>0$, output $\lerror$, and all nodes $p_{i'}$ with $i'<i$ output $\lstart(\phi)$, then the verifier rejects on $p_i$. Therefore, let us assume that there is at least a node that outputs a specific error $\lerror^x$. We write $\successor(v)$ and $\ldist(u,v)$ to denote respectively the successor of a node $v$ in the path, and the distance between two nodes $u$ and $v$ in the path.
\begin{itemize}
	\item If $x = 0$, the verifier accepts only if this error produces a chain that starts from $p_0$ and proceeds with increasing values. In order to be accepted, this chain must end at a node $w'$, and $\successor(w')$ must output $\lerror$. Then, $\successor(w')$ must witness that $w'$ indeed has a local error in the machine initialization, which is not possible in a good input.
	
	\item If $x = 1$, we could have two cases:
	\begin{itemize}
		\item there is a chain of increasing values that starts from a node $w$ with $\linput(w)=\lseparator$, and ends on a node $w'$ such that $\ldist(w,w')<B$, $\linput(\successor(w'))=\lseparator$, and $\loutput(\successor(w'))=\lerror$ (the tape is too short);
		\item there is a chain of increasing values that starts from a node $w$ with $\linput(w)=\lseparator$, and ends on a node $w'$ such that $\ldist(w,w')=B$, $\linput(\successor(w'))\neq\lseparator$, and $\loutput(\successor(w'))=\lerror$ (the tape is too long).
	\end{itemize}
	Since, in a good input, the distance between two nodes having in input \lseparator{} is always $B+1$, the above scenarios are not possible.
	
	\item If $x = 2$, there must be a chain of length exactly $B+1$, starting from a node $w$ having $\linput(w)=\tape(c,s,h)$, where $c=x\in\{0,1\}$, $h=\mbox{false}$, and ending on a node $w'$ such that $\ldist(w,w')=B+1 $, and $\linput(w')=\tape(c',s',h')$, where $c'\neq x$. In a good input, the tape content of nodes $w$ and $w'$ must be the same.
	
	\item If $x=3$, it means that there must exist two neighbors having two different states, and this can not happen in a good input.
	
	\item If $x=4$, there must be a chain that propagates the old state and old input, and the verifier accepts only if acknowledges that the transition has been wrongly encoded, which can not the case in a good input.
	
	\item If $x=5$, there must be a chain of length at least $2$ not passing through nodes having in input \lseparator, starting from a node $w$ with $\linput(w)=\tape(c,s,h)$ where $h=\mbox{true}$, and ending on a node $w'$ with $\linput(w')=\tape(c',s',h')$ where $h'=\mbox{true}$. This is not possible on a good input.
\end{itemize}
Therefore, since nodes can not output any kind of error, and since $\lempty$ is not a valid output for the nodes encoding the \lba{}, then these nodes must output $\lstart(\phi)$, where the value of $\phi$ matches the input of the first node of the path. Hence, $\Pi_{M_B}$ requires $\Omega(T'')$.

\subsection{\boldmath Normalizing an \texorpdfstring{$\lcl{}$}{LCL} Problem}
We now show how to $\beta$-normalize an \lcl{} $\Pi_{M_B}$ and obtain a new \lcl{} having roughly the same time complexity. We define three different verifiers depending on their view.
\begin{itemize}[noitemsep]
	\item A $\mathcal V_{\operatorname{in,in}-\operatorname{out,out}}$ verifier running at node $v$, checks $\linput(v)$, $\linput(\predecessor(v))$, $\loutput(v)$, and $\loutput((\predecessor(v))$.
	\item A $\mathcal V_{\operatorname{in}-\operatorname{out}}$ verifier running at node $v$, checks $\linput(v)$ and $\loutput(v)$.
	\item A $\mathcal V_{\operatorname{out}-\operatorname{out}}$  verifier running at node $v$ checks $\loutput(v)$ and $\loutput((\predecessor(v))$.
\end{itemize}

\begin{lemma}
	Consider an \lcl{} $\PP$ with $|\LabelIn^{\PP}| = \alpha$ and $|\LabelOut^{\PP}| = \beta$ that can be solved in time $T$ and can be locally checked with a $\mathcal V_{\operatorname{in,in}-\operatorname{out,out}}$ verifier. It is possible to define an \lcl{} $\PP'$ such that $|\LabelIn^{\PP'}| = \alpha$ and $|\LabelOut^{\PP'}| = \alpha \cdot \beta$ that can be solved in time $T$ and can be locally checked with a $\mathcal V_{\operatorname{in}-\operatorname{out}}$ and a $\mathcal V_{\operatorname{out}-\operatorname{out}}$ verifier.
\end{lemma}

\begin{proof}
	We define $\Sigma^{\PP'}_{\operatorname{in}} = \Sigma^{\PP}_{\operatorname{in}}$, and $\Sigma^{\PP'}_{\operatorname{out}} = \Sigma^{\PP}_{\operatorname{out}} \times \Sigma^{\PP}_{\operatorname{out}}$. Let $\loutput(v) = (\lin, \lout) \in \Sigma^{\PP'}_{\operatorname{out}}$. Let $\loutput(\predecessor(v)) = (\lin',\lout') \in \Sigma^{\PP'}_{\operatorname{out}}$. The $\mathcal V_{\operatorname{in}-\operatorname{out}}$ verifier checks that $\linput(v) = \lin$. The $\mathcal V_{\operatorname{out}-\operatorname{out}}$ verifier acts the same as the $\mathcal V_{\operatorname{in,in}-\operatorname{out,out}}$ verifier executed on $((\lin',\lout'),(\lin,\lout))$. The \lcl{} problem $\PP'$ can be solved with the  following algorithm at each node $v$.
	\begin{itemize}[noitemsep]
		\item Gather the ball $B_v(T)$.
		\item Simulate the original algorithm on  $B_v(T)$; let $\lout$ be the output of this simulation.
		\item Output $(\linput(v),\lout)$.
	\end{itemize}
	It is easy to check that this output is valid for the problem $\PP'$, and that it requires $T$ rounds. Also, note that it is not possible to solve $\PP'$ faster than $T$. In fact, in order to satisfy the $\mathcal V_{\operatorname{in}-\operatorname{out}}$ verifier, the input must be copied correctly; while in order to satisfy the $\mathcal V_{\operatorname{out}-\operatorname{out}}$ verifier, we need to satisfy the $\mathcal V_{\operatorname{in,in}-\operatorname{out,out}}$ verifier executed giving the same input that it would have seen on $\PP$.
\end{proof}

\begin{lemma}
	Consider an \lcl{} $\PP$ with $|\Sigma^{\PP}_{\operatorname{in}}| = \alpha$ and $|\Sigma^{\PP}_{\operatorname{out}}| = \beta$ that can be solved in time $T$ and can be locally checked with a $\mathcal V_{\operatorname{in}-\operatorname{out}}$ and a $\mathcal V_{\operatorname{out}-\operatorname{out}}$ verifier. We can define a $\beta'$-normalized \lcl{} $\PP'$ with $\beta'=|\Sigma^{\PP'}_{\operatorname{out}}| = 2^{\gamma}\cdot (|\Sigma^{\PP}_{\operatorname{out}}| +3)$ that can be solved in time $\Theta(\gamma \cdot  T(n/\gamma))$, where $\gamma = 2\lceil\log\alpha \rceil +3$.
\end{lemma}
\begin{proof}
	In the following we will exploit the ability of an algorithm to work on identifiers that can be polynomial in the size of the graph. In particular, we assume that if an algorithm works on an instance with IDs in the range $1,\ldots,r$, then it works also on an instance with IDs in the range $1,\ldots,\gamma\cdot r$. Then, we show how to define an \lcl{} $\PP'$ such that:
	\begin{itemize}[noitemsep]
		\item if the input instance encodes a virtual instance for the problem $\PP$, it is required to solve $\PP$ on the virtual instance;
		\item otherwise, it is required to prove that the encoding is wrong.
	\end{itemize}
	Let $\mathcal V'_{\operatorname{in}-\operatorname{out}}$ and $\mathcal V'_{\operatorname{out}-\operatorname{out}}$ be the verifiers of our $\beta'$-normalized \lcl.
	\paragraph{\boldmath Encoding $\PP$ in $\PP'$.} We start by defining how to encode an instance of $\PP$ of size $n$, as an instance of $\PP'$ of size $N = \gamma \cdot n$. We denote with $p_0,\ldots,p_{n-1}$ and $p'_0,\ldots,p'_{\gamma \cdot n -1}$ respectively the instance of $\PP$ and the one of $\PP'$. Let $a= \lceil\log\alpha \rceil$. For the sake of simplicity, let us rename nodes $p'_{\gamma i},\ldots ,p'_{\gamma(i+1) -1}$, where $0\le i \le n-1$, as $q^i_0, \ldots, q^i_{2a+2}$ (notice that $2a+2=\gamma-1$).  The first $a+1$ nodes, $q^i_0,\ldots, q^i_{a}$, have input $1$, while nodes $q^i_{a+1}$ and $q^i_{2a+2}$ have input $0$. Each of the remaining $a$ nodes, $q^i_{a+2},\ldots,q^i_{2a+1}$, has in input one bit of the binary representation of $\linput(p_i)$, in some fixed order (see Figure \ref{normalization} for an illustration).
	
	\paragraph{\boldmath The $\mathcal V'_{\operatorname{in}-\operatorname{out}}$ Verifier.} The set of output labels of $\PP'$ is $\Sigma^{\PP'}_{\operatorname{out}} = 2^\gamma \times (\Sigma_{\operatorname{out}} \cup \{ E_r, E, E_l \})$. Let $q^i_j$ be a node of the instance of $\PP'$ where  $0\le i \le n-1$ and $0\le j \le 2a+2$. Let $\linput(q^i_j) \in \{0,1\}$, and let $\loutput(q^i_j) = ((b_0,\ldots,b_{2a+2}),\lout) \in \Sigma^{\PP'}_{\operatorname{out}}$. The $\mathcal V'_{\operatorname{in}-\operatorname{out}}$ verifier running at $q^i_j$ checks that
	\begin{itemize}[noitemsep]
		\item $\linput(q^i_j) = b_0$, and
		\item if $\lout \in \Sigma^{\PP}_{\operatorname{out}}$, then
		\begin{itemize}
			\item if all bits in $b_0,\ldots,b_{a}$ are $1$s, checks that the original $\mathcal V_{\operatorname{in}-\operatorname{out}}$ verifier accepts on $(x,\lout)$, where $x$ is obtained by recovering the input $p_i$ for the original algorithm from $b_{a+2},\ldots,b_{2a+1}$.
		\end{itemize}
	\end{itemize}
	
	\paragraph{\boldmath The $\mathcal V'_{\operatorname{out}-\operatorname{out}}$ Verifier.} Let the output of $p^i_j$ be $((b_0,\ldots,b_{2a+2}),\lout)\in \Sigma^{\PP'}_{\operatorname{out}}$ and the output of the predecessor of $p^i_j$ be $((b'_0,\ldots,b'_{2a+2}),\lout') \in \Sigma^{\PP'}_{\operatorname{out}}$
	The O-O verifier first checks that
	\begin{itemize}[noitemsep]
		\item $b_0 = b'_1, b_1=b'_2, \ldots, b_{2a+1} = b'_{2a+2}$, and
		\item if $\lout \in \Sigma^{\PP}_{\operatorname{out}}$ and $\lout' \in \Sigma^{\PP}_{\operatorname{out}}$, then
		\begin{itemize}
			\item if at least one bit in $b_0,\ldots,b_{a}$ is $0$ , then $\lout = \lout'$,
			\item if all bits in $b_0,\ldots,b_{a}$ are $1$s, then check that the original $\mathcal V_{\operatorname{out}-\operatorname{out}}$ executed on $(\lout',\lout)$ accepts.
		\end{itemize}
	\end{itemize}
	\paragraph{Dealing with Errors.} We now add some constraints to handle the case in which $\lout \notin \Sigma^{\PP}_{\operatorname{out}}$. Let the output of $p^i_j$ be $((b_0,\ldots,b_{2a+2}),\lout)$ and the output of the predecessor of $p^i_j$ be $((b'_0,\ldots,b'_{2a+2}),\lout')$. The $\mathcal V'_{\operatorname{in}-\operatorname{out}}$ verifier additionally checks that, if $\lout = E$, then the encoding is not locally valid, that is,
	\begin{itemize}[noitemsep]
		\item either
		\begin{itemize}
			\item there are two numbers $x,y\ge 0$, $x+y \le a$, such that $b_0,\ldots,b_{x-1}$ and $b_{2a+3-y},\ldots,b_{2a+2}$ are all equal to $1$, $b_x = 0$, $b_{2a+2-y} = 0$, and
			\item there is not a contiguous sequence of length $a+1$ of all $1$s in $b_0, \ldots, b_{2a+2}$,
		\end{itemize}
		\item or $b_0,\ldots,b_a$ are all $1$s but either $b_{a+1}\neq 0$ or $b_{2a+2}\neq 0$.
	\end{itemize}
	The $\mathcal V'_{\operatorname{out}-\operatorname{out}}$ verifier running on $p^i_j$ additionally checks that,
	\begin{itemize}[noitemsep]
		\item if $\lout=E_l$,  then $p^i_j$ must have a predecessor, and it must hold that $\lout'\notin \{E_r\} \cup \Sigma^{\PP}_{\operatorname{out}}$;
		\item if $\lout \in \Sigma^{\PP}_{\operatorname{out}}$, if $p^i_j$ has a predecessor, then $\lout'$ must be different from $E_r$;
		\item if $\lout=E_r$, then $p^i_j$ must have a successor.
	\end{itemize}
	Let $N$ be the size of the graph. An algorithm solving $\PP'$ in $(\gamma+1)\cdot T$ rounds does the following at each node $v'$:
	\begin{itemize}[noitemsep]
		\item Gather the ball $B_{v'}((\gamma+1)\cdot T)$.
		\item If $B_{v'}((\gamma+1)\cdot T)$ looks like a correct encoding of an input instance of $\PP$
		\begin{itemize}
			\item let $u'$ be the nearest left node having input $1$ and other $a$ successors, $\successor^{(1)} u',\ldots,\successor^{(a)} u'$, having also input $1$
			\item Compute the virtual instance for $\PP$ (setting the IDs to be the same of the nodes satisfying the above)
			\item Simulate the original algorithm on the virtual instance by setting $n=N/\gamma$, let $\lout$ be the output of $u'$
			\item Output $((\linput(v'),\linput(\successor^{(1)} v'),\ldots,\linput(\successor^{(\gamma-1)} v'),\lout)$
		\end{itemize}
		\item Otherwise,
		\begin{itemize}
			\item if there is a local error, output $E$
			\item if the nearest error is on the left, output $E_l$
			\item otherwise, output $E_r$
		\end{itemize}
	\end{itemize}
	It is easy to check that the output of the algorithm satisfies the constraints. In order to show a lower bound for the new \lcl{}, we now show that it is not possible to produce errors in a graph that is a valid encoding.
	In fact, nodes can not cheat by wrongly outputting the input of the neighbors, otherwise either the input-output verifier notices inconsistencies on the first bit, or the output-output verifier notices inconsistencies on the other bits. Then, on a valid encoding, no input satisfies the constraints that allows to produce $E$ as output. Finally, the constraints impose that a chain of $E_r$ or $E_l$ points to a node that is outputting $E$.
	
	Note that, if the original \lcl{} has complexity $T$, then the new \lcl{}, on instances of size $N = \gamma \cdot n$, has complexity $\Theta(\gamma(T(n))) = \Theta(\gamma \cdot T(\frac{N}{\gamma}))$.	
\end{proof}

\subsection{Hardness Results}

\begin{theorem}\label{betahard}
	There are $\beta$-normalized \lcl{}s that can be solved in constant time but the distributed time complexity is $2^{\Omega(\beta)}$.
\end{theorem}
\begin{proof}
	The complexity of $\Pi_{M_B}$ is $\Theta(B \cdot T)$ if $M_B$ terminates in $T$ steps. $|\Sigma_{\operatorname{in}}| = O(1)$ and $|\Sigma_{\operatorname{out}}| = \Theta(B)$. We can convert it to an \lcl{} where $|\Sigma_{\operatorname{in}}| = 2$, $|\Sigma_{\operatorname{out}}| = \Theta(B)$, and the complexity is still $\Theta(B \cdot T)$. There exist \lba{}s that terminate in $2^{\Theta(B)}$ steps (e.g. a binary counter). Thus, the complexity of the obtained \lcl{} is $\Theta(B \cdot 2^{\Theta(B)})$, that is $2^{\Omega(B)} = 2^{\Omega(|\Sigma_{\operatorname{out}}|)}=2^{\Omega(\beta)}$.
\end{proof}

\begin{theorem}\label{pspacehard}
	It is PSPACE-hard to distinguish whether a given \lcl{} problem $\PP$ with input labels can be solved in $O(1)$ time or needs $\Omega(n)$ time on globally oriented path graphs.
\end{theorem}
\begin{proof}
	It is PSPACE-hard to distinguish whether a given \lba{} terminates or loops (see e.g.~\cite{Esparza98}). Note that the description of a $\beta$-normalized \lcl{} has size $O(\beta ^2)$. In order to decide if a $\beta$-normalized version of a problem in $\Pi$ requires $O(1)$ or $\Omega(n)$ we need to decide if its associated \lba{}, running on a tape of size $B = \Theta(\beta)$, terminates or loops, and this implies the theorem.
\end{proof}

\subsection{Extending the Results to Undirected Cycles}
We show how to extend the above results, which apply to globally oriented paths, to the case where the input graph is an undirected path or an undirected cycle. We first focus on showing how to adapt these results to undirected paths. Given a $\beta$-normalized \lcl{} $\PP$ defined on directed paths, we can define an \lcl{} $\PP'$ in which the set of input labels is $\Sigma'_{\operatorname{in}} = \{0,1\} \times \{0,1,2\}$, and the set of output labels is $\Sigma'_{\operatorname{out}} = \{0,1,2\} \times \{1,\ldots,\beta,E\}$. Let $\mathcal V_{\operatorname{in}-\operatorname{out}}$ and $\mathcal V_{\operatorname{out}-\operatorname{out}}$ be the verifiers of the $\beta$-normalized \lcl{} $\PP$, and let $\mathcal V'_{\operatorname{in}-\operatorname{out}}$ and $\mathcal V'_{\operatorname{out}-\operatorname{out}}$ be the verifiers of the \lcl{} $\PP'$. The idea is that we can use $3$ symbols to give an orientation as input to the nodes, by giving $0$ to the first node, $1$ to the second, $2$ to the third, $0$ to the fourth, and so on. Nodes must copy their orientation number to the output, and then, if the given orientation is consistent, nodes are required to solve the original problem $\PP$. On the other hand, if the orientation is not consistent, nodes are allowed to output an error $E$. Also, in order to avoid the need of error pointers, we allow nodes to treat the places where  the orientation is not consistent, as a place where the path ends.

This new \lcl{} can be checked as follows. The $\mathcal V'_{\operatorname{in}-\operatorname{out}}$ verifier takes in input the input and the output of the current node (as before) and first checks that the orientation has been copied correctly, and then checks that the original $\mathcal V_{\operatorname{in}-\operatorname{out}}$ verifier accepts. To verify the output, we allow the $\mathcal V'_{\operatorname{out}-\operatorname{out}}$ verifier to see slightly more than the original verifier $\mathcal V_{\operatorname{out}-\operatorname{out}}$. The $\mathcal V'_{\operatorname{out}-\operatorname{out}}$ verifier sees a triple containing the output of the node and the outputs of its neighbors. Note that the verifier does not know the orientation of the path (and the orientation of the triple), but it can recover it from the output of the nodes (that contains a copy of the orientation given as input). Then the $\mathcal V'_{\operatorname{out}-\operatorname{out}}$ verifier checks that, if the node outputted $E$, the orientation is indeed wrong. If the output is a value in $1,\ldots,\beta$, the verifier $\mathcal V'_{\operatorname{out}-\operatorname{out}}$ runs the original $\mathcal V_{\operatorname{out}-\operatorname{out}}$  verifier, since it can compute which neighbor is the predecessor. It is easy to see that the complexity of $\PP'$ is the same as the one of $\PP$.

The \lcl{} description of $\PP'$, that is, the size of its input, its output, and its verifier, is now $O(\beta^3)$, therefore the hardness result still applies.

We now show how an \lcl{} for paths can be converted to an \lcl{} for cycles.  The idea is the following. On a cycle, we give an additional input to each node in $\{0,1\}$. Nodes marked as $1$ are exempt to solve the problem and act as separators between the other nodes. That is, nodes are required to solve the original problem on the subpaths that lie between nodes marked as $1$. It may be the case that no node has $1$ as input. In this case we allow nodes to output a special error. If a node decides to output this error, both its neighbors must output the same, that is, \emph{all} nodes must output the same. We impose the constraint that a node marked as $1$ can not output this error. A worst case instance would be the one in which one node is marked $1$ and all other nodes are marked $0$---this would represent a path with a length that is roughly equal to the one of the cycle. There is one case that requires a bit of care: if all nodes are marked $0$, but the original problem can be solved in sublinear time, nodes could not be able to coordinate to produce the special error. For our purpose, it is possible to check that, if we consider the problem $\Pi_{M_B}$ previously defined in the case in which $M_B$ terminates, in an instance in which nobody has a predecessor, nodes can efficiently solve the problem by just outputting $\lerror$ or $\lempty$, depending on their input.

\subsection{Encoding Input Labels as Trees \label{sec.encode}}
In this section we demonstrate a reduction from the \lcl{} problem  $\mathcal{P}$ with input labels on {\em any} graph $G$ to an \lcl{} problem $\mathcal{P}^\star$ without input labels on the modified graph $G^\star$. The modified graph $G^\star$ is the result of attaching a rooted tree to each $v \in V(G)$ that encodes the input label of $v$ for the \lcl{} problem  $\mathcal{P}$. The reduction  allows us to extend the hardness proof to the case of \lcl{} problems without input labels.

\paragraph{Encoding.} Given a $2^k$-bit binary string $S=(s_1, \ldots, s_{2^k})$, define $\Enc(S)$ as the rooted tree constructed as follows.
\begin{itemize}
	\item Begin with the full binary tree which has $2^k$ leaves, and the distance from the root to each leaf is $k$.
	\item For each non-leaf node $v$, let $u$ be any one of its two children, and subdivide the edge $\{v,u\}$ into two edges $\{v,w\}$ and $\{w,u\}$, where $w$ is a new node. The node $w$ is designated as the left child of $v$.
	\item Let $U=(u_1 \ldots, u_{2^k})$ be the leaves ordered by the in-order traversal.
	\item For each $i \in [2^k]$, add two new nodes $x$ and $y$ as the children of  $u_i$. If $s_i = 1$, add  two more new nodes $x'$ and $y'$ and the two edges $\{x, x'\}$ and $\{y, y'\}$.
\end{itemize}
The tree $\Enc(S)$ has maximum degree 3, and all nodes are within distance $2(k + 1)$ to the root. Given a graph $G$ such that each node $v \in V(G)$ is associated with an input label $L(v) \in \LabelIn$, define $G^\star$ as the graph resulting from the following operations on $G$. For each node $v \in V(G)$, attach the rooted tree $T = \Enc(L(v))$ to $v$ by adding the edge $\{v,z\}$, where $z$ is the root of $T$. Notice that  $\Delta(G^\star) = \max\{3, \Delta(G) +1\}$.

\paragraph{Decoding.} Given a rooted tree $T=\Enc(S)$ for some $S \in \{0,1\}^{2^k}$, define $\Dec(T) = S$. The decoding can be done by the following procedure. Consider an in-order traversal of the tree such that (i) for each node $v$ such that exactly one of its children $w$ has degree 2, treat $w$ as the left child of $v$, (ii) print `1' if a node that has two children of degree 2 is encountered, (iii) print `0' if a node that has two children of degree 1 is encountered. Then the printed sequence is $S$.

\paragraph{\boldmath The Modified $\lcl$ Problem $\mathcal{P}^\star$.} Let $\mathcal{P}$ be an $\lcl$ problem  with input labels. Suppose that the radius of $\mathcal{P}$ is $r$, and the maximum degree is $\Delta$. Set $k = \lceil \log  \log |\LabelIn| \rceil$, and let each label in $\LabelIn$ be represented by a distinct $2^k$-bit string. The modified $\lcl$ problem $\mathcal{P}^\star$, which does not require input label, is defined by the following rules. The set of the output labels of $\mathcal{P}^\star$ is $\LabelOut$, the same as that of $\mathcal{P}$. Let $G^\star$ be a graph with maximum degree $\leq \Delta +1$.
\begin{itemize}
	\item Define $G_1 = G^\star$.
	
	For each $2 \leq i \leq k+2$, define $G_i$ as the graph induced by nodes in $V(G^\star) - \bigcup_{j = 1}^{i-1} (A_j \cup B_j)$.
	
	For each $i \in [k+2]$, define $A_i = \{v \in V(G_i) | \deg_{G_i}(v) = 1\}$.
	
	For each $i \in [k+1]$, define $B_i = \{v \in V(G_i) | \deg_{G_i}(v) = 2 \text{ and } \exists u \in A_i \text{ s.t. } \{u,v\} \in E(G_i)\}$.
	\item Define $V_{\text{label}} = \bigcup_{j = 1}^{k+2} A_j \cup \bigcup_{j = 1}^{k+1}  B_j$.
	
	Define  $V_{\text{main}}$ as the set of nodes in $V(G^\star) \setminus V_{\text{label}}$ that have exactly one neighbor in $V_{\text{label}}$.
	\item For each $v \in V_{\text{main}}$, define $L(v)$ as follows. Let $u$ be the unique node in $V_{\text{label}}$ adjacent to $v$, and let $T$ be the connected component induced by nodes in $V_{\text{label}}$ that contains $u$. Set $L(v) = \Dec(T)$ (with $u$ being the root of $T$). If the decoding procedure $\Dec(T)$ fails, simply set $L(v)$ as the {\em first} label in $\LabelIn$.
	\item The output labeling, together with the input labeling defined by the function $L(\cdot)$, forms a legal labeling of the subgraph induced by the nodes in $V_{\text{main}}$ for $\mathcal{P}$.
\end{itemize}
Notice that the connected components in $V_{\text{label}}$ are the trees encoding input labels,
and the subgraph induced by the nodes in $V_{\text{main}}$ is $G$ (as long as $G$ does not contain isolated node).
The function $L(v)$ recovers the input label of $v$ for each $v \in V(G)$.
Given that $\mathcal{P}$ has a valid labeling on {\em all} graphs (resp., trees) of maximum degree $\Delta$, the modified $\lcl$ problem $\mathcal{P}^\star$ also has a valid labeling on all graphs (resp., trees) of maximum degree $\Delta + 1$.

\paragraph{Reducing the Radius.} The above definition of $\mathcal{P}^\star$ requires radius $r + O(k)$, as a node $v \in V_{\text{main}}$ needs $O(k)$ extra rounds to calculate $L(u)$ for all $u \in N^r(v) \cup V_{\text{main}}$. We present a simple modification that reduces the radius to only $r + O(1)$ at the cost of expanding the number of output labels from $|\LabelOut|$ to $|\LabelOut| + 2^{2^k} < |\LabelOut| + |\LabelIn|^2$. The idea is to let nodes in $V_{\text{label}}$ to use output labels to pass the information stored at the leaves to the root based on local rules. Consider a connected component $T$ induced by nodes in $V_{\text{label}}$. The subgraph $T$ is interpreted as a tree rooted at the unique node in $T$ that is adjacent to some node in $V_{\text{main}}$.
\begin{description}
	\item[Base Case:] Let $v \in A_2$. If $v$ is adjacent to two nodes in $B_1$, the output label of $v$ is $1$; otherwise the output label of $v$ is $0$.
	\item[Root / Degree-3 Nodes:] Let $3 \leq i \leq k+2$, and let $v \in A_i$. Then $v$ has a unique neighbor $u_{\text{left}} \in B_{i+1}$ and a unique neighbor $u_{\text{right}} \in A_{i+1}$. Let $S_{\text{left}}$ be the output label of $u_{\text{left}}$, and let $S_{\text{right}}$ be the output label of $u_{\text{right}}$. Then the output label of $v$ is the binary string $S_{\text{left}} \circ S_{\text{right}}$.
	\item[Degree-2 Nodes:] Let $3 \leq i \leq k+1$, and let $v \in B_i$. Then $v$ has a unique neighbor $u \in A_{i+1}$. The output label of $v$ is the same as the output label of $u$.
\end{description}
Thus, for each node $v \in V_{\text{main}}$, $L(v)$ is simply the output label of the unique node $u \in V_{\text{label}} \cap N(v)$.

\begin{theorem}
	For any $\lcl$ problem $\mathcal{P}$ on any graph $G$ of maximum degree $\Delta$ that does not have isolated nodes, the following two statements are equivalent.
	\begin{itemize}
		\item The labeling $\LL\colon V(G) \rightarrow \LabelOut$ is a valid labeling of $G$ for the problem  $\mathcal{P}$.
		\item There exists some labeling $\LL'$ of the nodes in $V_{\text{label}}$ such that $\LL$ and $\LL'$ together form a valid labeling of $G^\star$ for the problem $\mathcal{P}^\star$.
	\end{itemize}
\end{theorem}

\begin{theorem}\label{deciding_deg3trees}
	It is PSPACE-hard to distinguish whether a given \lcl{} problem $\PP$ without input labels can be solved in $O(1)$ time or needs $\Omega(n)$ time on trees with degree $\Delta = 3$.
\end{theorem}

\section{Decidability}\label{sec.cycle}
In this section, we show that the two gaps $\omega(1)$---$o(\log^\ast n)$ and $\omega(\log^\ast n)$---$o(n)$ for \lcl\ problems {\em with input labels}  on paths and cycles are decidable. More specifically, given a specification of an \lcl\ problem $\mathcal{P}$, there is an algorithm that outputs a description of an asymptotically optimal deterministic $\LOCAL$ algorithm for $\mathcal{P}$, as well as its time complexity.

We will prove the statements for the case of cycles, but the analogous results for cycles and paths follows as a simple corollary, as we can encode constraints related to degree-$1$ nodes as constraints related to nodes adjacent to a special input label. Furthermore, having a promise that the input is a path does not change the time complexity of an $\lcl$ problem: if a problem can be solved in time $T = o(n)$ in labeled paths, the same algorithm will solve it also in time $T = o(n)$ in labeled cycles.

The proof of Theorem~\ref{thm-cycle-gap-2} is in Section~\ref{sect-cycle-gap-logstar}; the proof of Theorem~\ref{thm-cycle-gap-1} is in Section~\ref{sec.cycle-main}.

\begin{theorem}\label{thm-cycle-gap-2}
	For any \lcl\ problem $\mathcal{P}$ on cycle graphs, its  deterministic $\LOCAL$ complexity is either  $\Omega(n)$ or  $O(\log^\ast n)$. Moreover, there is an algorithm that decides whether $\mathcal{P}$ has complexity  $\Omega(n)$ or  $O(\log^\ast n)$ on cycle graphs; for the case the complexity is $O(\log^\ast n)$, the algorithm outputs a description of an $O(\log^\ast n)$-round deterministic $\LOCAL$ algorithm that solves $\mathcal{P}$.
\end{theorem}

\begin{theorem}\label{thm-cycle-gap-1}
	For any \lcl\ problem $\mathcal{P}$ on cycle graphs, its deterministic $\LOCAL$ complexity is either  $\Omega(\log^\ast n)$ or  $O(1)$. Moreover, there is an algorithm that decides whether $\mathcal{P}$ has complexity  $\Omega(\log^\ast n)$ or  $O(1)$ on cycle graphs; for the case the complexity is $O(1)$, the algorithm outputs a description of an $O(1)$-round deterministic $\LOCAL$ algorithm that solves $\mathcal{P}$.
\end{theorem}	

\paragraph{Graph Notation.} For convenience, in this section, a directed path $P$ with input labels is alternatively described as a string in $\LabelIn^k$, where $k > 0$ is the number of nodes in $P$. Similarly, an output labeling $\LL$ of $P$ is alternatively described as a string in $\LabelOut^k$. In subsequent discussion, we freely switch between the graph-theoretic notation and the string notation.
Given an output labeling $\mathcal{L}$ of $P$, we say that $\mathcal{L}$ is {\em locally consistent} at $v$ if the input and output labeling assigned to $N^r(v)$ is acceptable for $v$. Note that $N^r(v)$ refers to the radius-$r$ neighborhood of $v$.  Given two integers $a \leq b$, the notation $[a, b]$ represents the set of all integers $\{a, a+1, \ldots, b\}$.  Given a string $w$, denote $w^R$ as the reverse of $w$.

\paragraph{Overview.}  Before we proceed, we briefly discuss the high level idea of the proofs. 
The main tool underlying the proofs is the ``pumping lemma'' which was developed in~\cite{chang17hierarchy}. Intuitively, we classify the set of all input-labeled paths into a finite number of equivalence classes satisfying the following property.  Let $P$ be a subpath of $G$, and let $P'$ be another path that is of the same equivalence class as $P$. Given a complete legal labeling of $G$, if we let $G'$ be the result of replacing $P$ with $P'$, then it is always possible to extend this partial labeling of $G'$ to a complete legal labeling by appropriately labeling $P'$.
The pumping lemma guarantees that for any path  $P$ whose length is at least the pumping constant $\Lpump$, and for any number $x \geq \Lpump$, there is another path $P'$ of length at least $x$ and $P'$ is of the same equivalence class as $P$.

Informally, in the proof of Theorem~\ref{thm-cycle-gap-2}, we show that any \lcl\ problem $\mathcal{P}$ solvable in $o(n)$ rounds can be solved in $O(\log^\ast n)$ rounds in the following canonical way based on a ``feasible labeling function'' $f$.  Intuitively, a labeling function $f$ is feasible if  for any given independent set $I$ that is sufficiently well-spaced, we can apply $f$ to assign the output labels to each $v \in I$ and its nearby neighbors locally such that this partial labeling can always be extended to a complete legal labeling.
The $\omega(\log^\ast n)$---$o(n)$ gap and the decidability result follows from these two claims.
\begin{itemize}
\item If there is an $o(n)$-round algorithm $\mathcal{A}$ that solves $\mathcal{P}$, then a feasible function $f$ exists. This is proved by first create an imaginary graph where some paths are extended using pumping lemmas, and then apply a simulation of $\mathcal{A}$ on the imaginary graph.
\item  Whether a feasible function exists is decidable. Intuitively, this is due to the fact that the number of equivalence classes is finite.
\end{itemize}

The proof of Theorem~\ref{thm-cycle-gap-1} is a little more complicated since the time budget is only $O(1)$, so we cannot even afford to find an MIS. 
To solve this issue, we decompose the cycle graph $G$ into paths with unrepetitive patterns and paths with repetitive patterns, in $O(1)$ rounds. For paths with unrepetitive patterns, we are able to compute a sufficiently well-spaced MIS in $O(1)$ rounds by making use of the irregularity of the input patterns. Paths with repetitive patterns are similar to the paths without input labels, and we will show that we can always label them by repetitive output patterns, given that the underlying \lcl\ problem is $o(\log^\ast n)$-time solvable.

\subsection{Pumping Lemmas for Paths\label{sec.pump}}

Let $P = (s, \ldots, t)$ be a directed path, where each node has an input label from $\LabelIn$.
The {\em tripartition} of the nodes $\xi(P)=(D_1,D_2,D_3)$ is defined as follows:
\begin{align*}
	D_1 		&= N^{r-1}(s) \cup N^{r-1}(t),\\
	D_2 		&= \left(N^{2r-1}(s)  \cup N^{2r-1}(t)\right) \setminus D_1,\\
	D_3 &= P \setminus (D_1 \cup D_2).
\end{align*}
See Figure~\ref{fig:tripartition} for an illustration. More specifically, suppose $P = (u_1, \ldots, u_k)$, and let $i \in [1,k]$.
Then we have:
\begin{itemize}[noitemsep]
	\item  $u_i \in D_1$ if and only if $i \in [1,r] \cup [k-r+1, k]$.
	\item  $u_i \in D_2$ if and only if $i \in [r+1, 2r] \cup [k-2r+1, k-r]$.
	\item  $u_i \in D_3$ if and only if $i \notin [1, 2r] \cup [k-2r+1, k]$.
\end{itemize}

Let $\LL \colon D_1 \cup D_2 \rightarrow \LabelOut$ assign output labels to $D_1\cup D_2$.
We say that $\LL$ is {\em extendible} w.r.t.\ $P$ if there exists a complete labeling
$\LL_\diamond$ of $P$ such that $\LL_\diamond$
agrees with $\LL$ on $D_1\cup D_2$, and $\LL_\diamond$ is locally consistent at all nodes
in $D_2\cup D_3$.

\begin{figure}
	\centering
	\includegraphics[width=0.6\textwidth]{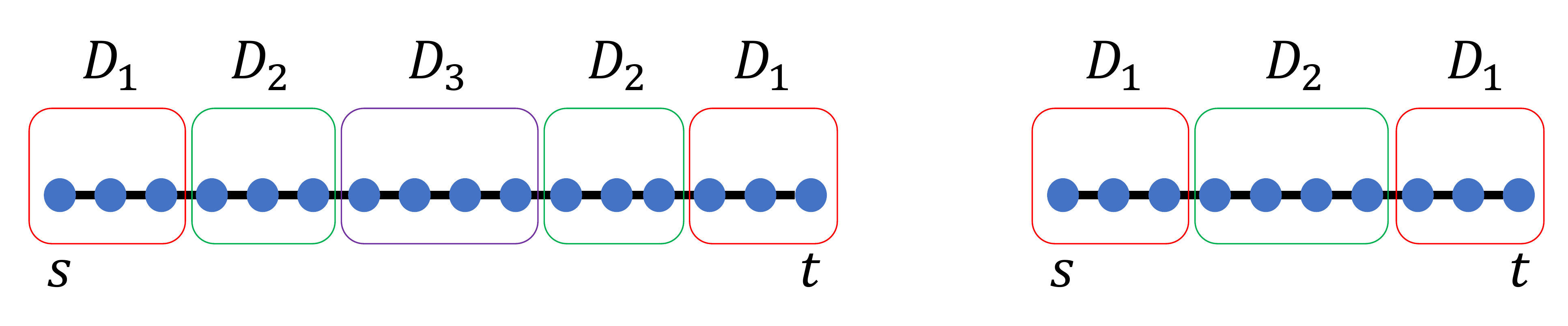}
	\caption{Illustration of the tripartition $\xi(P)=(D_1,D_2,D_3)$ with $r=3$.}\label{fig:tripartition}
\end{figure}

\paragraph{An Equivalence Class.}
We define an equivalence class $\simm$ for the directed paths (i.e., the set of all non-empty strings in $\LabelIn^{\ast}$), as follows.

\begin{oframed}
	\noindent
	Consider two directed paths $P=(u_1, \ldots, u_x)$ and $P' = (v_1, \ldots, v_y)$, and let $\xi(P)=(D_1,D_2,D_3)$ and  $\xi(P')=(D_1',D_2',D_3')$.
	Consider the following natural 1-to-1 correspondence  $\phi \colon (D_1 \cup D_2) \rightarrow (D_1'\cup D_2')$ defined as  $\phi(u_i) = v_i$ and $\phi(u_{x - i + 1}) = v_{y - i + 1}$ for each $i \in [1, 2r]$.
	The 1-to-1 correspondence  is well-defined so long as (i) $x = y$  or (ii) $x \geq 4r$ and  $y \geq 4r$.
	We have $P \simm P'$ if and only if the following two statements are met:
	
	\begin{itemize}
		\item \textbf{Isomorphism:} The 1-to-1 correspondence $\phi$ is well-defined, and for each $u_i \in D_1 \cup D_2$, the input label of $u_i$ is identical to the input label of $\phi(u_i)$.
		
		\item \textbf{Extendibility:} Let $\LL$ be {\em any} assignment of output labels to nodes in $D_1 \cup D_2$,
		and let $\LL'$ be the corresponding output labeling of $D_1' \cup D_2'$ under $\phi$.
		Then $\LL$ is extendible w.r.t.\ $P$ if and only if $\LL'$ is extendible w.r.t.\ $P'$.
	\end{itemize}
	Note that for the special case of $x \leq 4r$, we have $P \simm P'$  if and only if $P$ is identical to $P'$.
\end{oframed}

Define $\type(P)$ as the equivalence class of $P$ w.r.t.\ $\simm$.
The following technical lemma is analogous to~\cite[Lemma~1]{chang17hierarchy} in a specialized setting. We only use this lemma to prove the lemmas in Section~\ref{sec.pump}.

\begin{lemma}\label{lem:base-analogousss}
	Let $G$ be a path graph or a cycle graph where all nodes have input labels from $\LabelIn$.
	Let $P$ be a directed subpath of $G$, and let $P'$ be another directed path such that $\type(P')=\type(P)$.
	We write $\xi(P)=(D_1,D_2,D_3)$ and  $\xi(P')=(D_1'
	,D_2',D_3')$.
	Let $\LL_\diamond$ be any complete labeling of $G$ such that $\LL_\diamond$ is locally consistent at all nodes in $D_2 \cup D_3$.
	Let $G' = \replace(G,P,P')$ be the graph resulting from replacing $P$ with $P'$ in $G$.
	Then there exists a complete labeling $\LL_\diamond'$ of $G'$ such that the following two conditions are met.
	\begin{enumerate}
		\item For each $v \in \left(V(G) \setminus V(P)\right) \cup  (D_1 \cup D_2)$ and its corresponding $v' \in \left(V(G') \setminus V(P')\right) \cup  (D_1' \cup D_2')$, we have $\LL_\diamond(v) = \LL_\diamond'(v')$. Moreover, if $v  \in \left(V(G) \setminus V(P)\right) \cup  D_1$ and $\LL_\diamond$ is locally consistent at $v$, then $\LL_\diamond'$ is locally consistent at $v'$.
		\item $\LL_\diamond'$ is locally consistent at all nodes in $D_2' \cup D_3'$.
	\end{enumerate}
\end{lemma}
\begin{proof}
	The labeling  $\LL_\diamond'(v')$ of $G'$ for each $v' \in \left(V(G') \setminus V(P')\right) \cup \left( D_1' \cup D_2'\right)$  is chosen ``naturally'' as follows.
	For each $v' \in V(G') \setminus V(P')$, we set  $\LL_\diamond'(v') = \LL_\diamond(v)$ for its corresponding node $v \in  V(G) \setminus V(P)$.
	For each $v' \in D_1' \cup D_2'$,  we set  $\LL_\diamond'(v') = \LL_\diamond(v)$ for its corresponding node $v \in D_1 \cup D_2$ such that $\phi(v) = v'$ in the definition of $\simm$.
	At this point, it is clear that if $v  \in \left(V(G) \setminus V(P)\right) \cup  D_1$ has a locally consistent labeling under $\LL_\diamond$, then its corresponding node $v' \in \left(V(G') \setminus V(P')\right) \cup  D_1'$ also has a locally consistent labeling under $\LL_\diamond'$, so Condition~1 holds.

	Now, the labeling $\LL_\diamond'$ is only undefined for nodes in $D_3'$. We show that we can complete the labeling in such a way that is locally consistent at all nodes in $D_2' \cup D_3'$.
	Denote $\LL$ as $\LL_\diamond$ restricted to $D_1 \cup D_2$. Since $\LL_\diamond$ is locally consistent at all nodes in $P$, the labeling $\LL$ is extendible w.r.t.\ $P$.
	Note that if we let $\LL'$ be $\LL_\diamond$ restricted to $D_1' \cup D_2'$, then according to the way we define  $\LL_\diamond'$, the two labeling  $\LL'$ and $\LL$ are identical under the 1-to-1 correspondence  $\phi$ specified in the definition of $\simm$. That is, for each $v' \in D_1' \cup D_2'$,  we have  $\LL'(v') = \LL(v)$ for its corresponding node $v \in D_1 \cup D_2$ such that $\phi(v) = v'$.
	Since $P \simm P'$, the labeling $\LL'$ must be extendible w.r.t.\ $P'$.
	That is, there is a way to assign $\LL_\diamond'(v')$ for each $v' \in D_3'$ such that all nodes in $D_2' \cup D_3'$ have locally consistent labelings under $\LL_\diamond'$, so Condition~2 holds.
\end{proof}

One useful consequence of this lemma is that if we start with a path or a cycle $G$ with a legal labeling, after replacing its subpath $P$ with another one $P'$ having the same type as $P$, then it is always possible to assign output labeling to $P'$ to get a legal labeling without changing the already-assigned output labels of nodes outside of $P'$.

\begin{restatable}{lemma}{lemreplacepath}\label{lem:replace-path}
	Let $G$ be a path graph or a cycle graph where all nodes have input labels from $\LabelIn$.
	Let $P$ be a directed subpath of $G$, and let $P'$ be another directed path such that $\type(P')=\type(P)$.
	Let $\LL_\diamond$ be complete labeling of $G$ that is locally consistent at all nodes in  $P$.
	Let $G' = \replace(G,P,P')$ be the graph resulting from replacing $P$ with $P'$ in $G$.
	Then there exists a legal labeling $\LL_\diamond'$ of $G'$ such that the following two conditions are met.
	\begin{enumerate}
		\item For each $v \in V(G) \setminus V(P)$ and its corresponding $v' \in V(G') \setminus V(P')$, we have $\LL_\diamond(v) = \LL_\diamond'(v')$. Moreover, if $\LL_\diamond$ is locally consistent at $v \in V(G) \setminus V(P)$, then $\LL_\diamond'$ is locally consistent at $v'$.
		\item $\LL_\diamond'$ is locally consistent at all nodes in $P'$.
	\end{enumerate}
\end{restatable}

\begin{proof}
	We write $\xi(P')=(D_1'
	,D_2',D_3')$.
	Condition~1 in this lemma is implied by Condition~1 in Lemma~\ref{lem:base-analogousss}.
	To see that Condition~2 in this lemma holds,
	note that in this lemma we additionally require that $\LL_\diamond$ is locally consistent at all nodes in $P$. Therefore, Condition~1 of Lemma~\ref{lem:base-analogousss} implies that $\LL_\diamond'$ is locally consistent at all nodes in $D_1'$. This observation, together with Condition~2 of  Lemma~\ref{lem:base-analogousss}, implies that $\LL_\diamond'$ is locally consistent at all nodes in $P'$.
\end{proof}

The following lemma is analogous to~\cite[Theorem~4]{chang17hierarchy} in a specialized setting. We only use this lemma in Section~\ref{sec.pump}.

\begin{restatable}{lemma}{lemtypepath}\label{lem:type-path}
	Let $P = (v_1, \ldots, v_k)$, and let $P' = (v_1, \ldots, v_{k-1})$.
	Let the input label of $v_k$ be $\alpha$.
	Then $\type(P)$ is a function of $\alpha$ and $\type(P')$.
\end{restatable}

\begin{proof}
	We prove the following stronger statement.
	Let $G$ be a directed path, and let $H$ be a directed subpath of $G$.
	Suppose $H'$ is another directed path satisfying  $\type(H) = \type(H')$.
	Let $G' = \replace(G,H,H')$ be the result of replacing $H$ with $H'$ in $G$.
	Then we claim that $\type(G) = \type(G')$.
	The lemma is a corollary of this claim.

	Consider the tripartitions $\xi(H)=(B_1,B_2,B_3)$,  $\xi(H')=(B_1',B_2',B_3')$,
	$\xi(G)=(D_1,D_2,D_3)$, and  $\xi(G')=(D_1',D_2',D_3')$.
	We write $B_0 = V(G) \setminus V(H)$ and $B_0' =  V(G') \setminus V(H')$.
	
	Let $\phi^\star$ be the natural 1-to-1 correspondence from  $B_0 \cup B_1 \cup B_2$ to $B_0' \cup B_1' \cup B_2'$.
	Note that $D_1 \cup D_2 \subseteq B_0 \cup B_1 \cup B_2$ and $D_1' \cup D_2' \subseteq B_0' \cup B_1' \cup B_2'$.
	Also, the 1-to-1 correspondence  between $D_1 \cup D_2$ and $D_1' \cup D_2'$ given by  $\phi^\star$ is exactly the 1-to-1 correspondence $\phi$ specified in the requirement of $G \simm G'$.

	Let $\LL\colon (D_1 \cup D_2)\rightarrow \LabelOut$ and let $\LL'$
	be the corresponding output labeling of $D_1'\cup D_2'$, under the 1-to-1 correspondence $\phi$.
	To show that $G \simm G'$,  all we need to do is show that  $\LL$ is extendible w.r.t.~$G$ if and only if $\LL'$ is extendible w.r.t.~$G'$.
	Since we can also write $G = \replace(G',H',H)$,
	it suffices to show just one direction, i.e., if $\LL$ is extendible then $\LL'$ is extendible.
	
	Suppose $\LL$ is extendible. Then there exists an output labeling $\LL_\diamond$ of $G$ such that (i) for each $v \in D_1 \cup D_2$, we have $\LL_\diamond(v) = \LL(v)$, and (ii) $\LL_\diamond$  is locally consistent at all nodes in $D_2 \cup D_3$.
	Since  $D_2 \cup D_3 \supseteq B_2 \cup B_3$, we can apply Lemma~\ref{lem:base-analogousss}, which shows that there exists a complete labeling $\LL_\diamond'$ of $G'$ such that the two conditions in Lemma~\ref{lem:base-analogousss} are met. We argue that this implies that $\LL'$ is extendible.
	We verify that (i) $\LL'(v') = \LL_\diamond'(v')$
	for each $v' \in D_1' \cup D_2'$, and
	(ii) $\LL_\diamond'$  is locally consistent at all nodes in $D_2' \cup D_3'$.
	\begin{itemize}
		\item Condition~1 of Lemma~\ref{lem:base-analogousss} guarantees that $\LL_\diamond(v) = \LL_\diamond'(\phi^\star(v))$ for each $v \in \left(V(G) \setminus V(H)\right) \cup (B_1 \cup B_2) =  B_0 \cup B_1 \cup B_2$ and its corresponding node $\phi^\star(v) \in B_0' \cup B_1' \cup B_2'$.
		Since $D_1' \cup D_2' \subseteq B_0' \cup B_1' \cup B_2'$, we have $\LL'(v') = \LL_\diamond'(v')$
		for each $v' \in D_1' \cup D_2'$.
		\item
		The fact that $\LL_\diamond$  is locally consistent at all nodes in $D_2 \cup D_3$, together with Condition~1 in Lemma~\ref{lem:base-analogousss}, guarantees that $\LL_\diamond'$ is locally consistent at all nodes in $(D_2' \cup D_3') \setminus B_3'$.
		Condition~2 in Lemma~\ref{lem:base-analogousss} guarantees that $\LL_\diamond'$ is locally consistent at all nodes in $B_2' \cup B_3'$. Therefore, $\LL_\diamond'$ is locally consistent at all nodes in $D_2' \cup D_3'$, as required. \qedhere
	\end{itemize}
\end{proof}

The number of types can be upper bounded as follows.

\begin{lemma}\label{lem-type-number}
	The number of equivalence classes of $\simm$ (i.e., types) is at most $|\LabelIn|^{4r} 2^{|\LabelOut|^{4r}}$.
\end{lemma}
\begin{proof}
	Let $P$ be a directed path, and let $\xi(P)=(D_1,D_2,D_3)$. Then $\type(P)$ is determined by the following information.
	\begin{itemize}
		\item The input labels in $D_1 \cup D_2$. Note that there are at most $|\LabelIn|^{4r}$ possible input labeling of $D_1 \cup D_2$.
		\item A length-$x$ binary string indicating the extendibility of each possible output labeling of $D_1 \cup D_2$, where $x = |\LabelOut|^{4r}$.
	\end{itemize}
	Therefore, the number of equivalence classes of $\simm$ is at most $|\LabelIn|^{4r} 2^{|\LabelOut|^{4r}}$.
\end{proof}

Define $\Lpump$ as the total number of types.
Observe that Lemma~\ref{lem:type-path} implies that $\type(P)$ can be computed by a finite automaton whose number of states is the total number of types, which is a constant independent of $P$.
Thus, we have the following two {\em pumping lemmas} which allow us to extend the length of a given directed path $P$ while preserving the type of $P$. The following two lemmas follow from the standard pumping lemma for regular language.

\begin{lemma}\label{thm:pump1}
	Let $P \in \LabelIn^k$ with $k \geq \Lpump$. Then $P$ can be decomposed into three substrings $P=x \circ y \circ z$ such that (i) $|xy| \leq \Lpump$, (ii) $|y|\geq 1$, and (iii) for each non-negative integer $i$, $\type(x \circ y^i \circ z) = \type(P)$.
\end{lemma}

\begin{lemma}\label{thm:pump2}
	For each $w \in \LabelIn^{>0}$, there exist two positive integers $a$ and $b$ such that $a + b \leq \Lpump$, and $\type(w^{ai + b})$ is invariant for each non-negative integer $i$.
\end{lemma}

\subsection{\boldmath The \texorpdfstring{$\omega(\log^\ast n)$}{omega(log* n)}---\texorpdfstring{$o(n)$}{o(n)} Gap}\label{sect-cycle-gap-logstar}

In this section we show that the $\omega(\log^\ast n)$---$o(n)$  gap is decidable. More specifically, we show that an \lcl\ problem $\mathcal{P}$ can be solved in $O(\log^\ast n)$ rounds if and only if there exists a {\em feasible function}, which is defined as follows.

\begin{description}
	\item[Input:] A directed path $P = w_1 \circ S \circ w_2$, where $|w_1| \in [\Lpump, \Lpump+1]$, $|w_2| \in [\Lpump, \Lpump+1]$,
	and $|S|=2r$. The decomposition $P = w_1 \circ S \circ w_2$ is considered part of the input.
	\item [Output:] A string $\mathcal{L} \in \LabelOut^{2r}$ that represents the output labeling of $S$.
	\item [Requirement:] Any such function $f$ is said to be {\em feasible} if the following requirement is met for any paths $S_1, S_2$ and $w_a, w_b, w_c, w_d$ such that $|S_1| = |S_2| = 2r$ and $\{|w_a|, |w_b|, |w_c|, |w_d|\} \subseteq [\Lpump, \Lpump+1]$.
	Let $P = w_a \circ S_1 \circ w_b \circ w_c \circ S_2 \circ w_d$, and
	consider the following assignment of output labels to $S_1 \cup S_2$.
	\begin{itemize}
		\item Either label $S_1$ by $f(w_a \circ S_1 \circ w_b)$ or label $S_1^R$ by $f(w_b^R \circ S_1^R \circ w_a^R)$.
		\item Either label $S_2$ by $f(w_c \circ S_2 \circ w_d)$ or label $S_2^R$ by $f(w_d^R \circ S_2^R \circ w_c^R)$.
	\end{itemize}
	It is required that given such a partial labeling of $P$, the middle part $w_b \circ w_c$ can be assigned output labels in such a way that the labeling of (i) the last $r$ nodes of $S_1$, (ii) all nodes in $w_b \circ w_c$, and (iii) the first $r$ nodes of $S_2$ are locally consistent.
\end{description}

The following lemma is a straightforward consequence of the well-known $O(\log^\ast n)$-round MIS algorithm on cycles.

\begin{lemma}\label{lem:ruling-set}
Let $G$ be a cycle graph of $n$ nodes, and let $s \leq k$ be two constant integers such that $s+k \leq n$.
Then in $O(\log^\ast n)$ rounds we can compute a decomposition $V = A \cup B$ such that each connected component of $A$ has size $s$, and each connected component of $B$ has size within $[k, k+1]$.
\end{lemma}
\begin{proof}
For any given constant integer $1 \leq L < n$, we will show that in $O(\log^\ast n)$ time we can find an independent set $I$ of $G$ such that each connected component induced by $V \setminus I$ has at least $L$ nodes and at most $2L$ nodes.
Using this result with $L = 2(s-1)+ k(s+k+1)$, it is straightforward to obtain the desired decomposition $V = A \cup B$, as follows.

For each $v \in I$, it arbitrarily chooses a size-$s$ path $S_v$ that contains $v$, and all nodes in $S_v$ are included to $A$.
Now each connected component $S'$ induced by the remaining nodes is a path of size at least $L -  2(s-1) \geq k(s+k+1)$ and at most $2L$.
We will divide the path $S'$ into subpaths $R_1, R_2, \ldots, R_t$ meeting the following conditions: (i) if $i$ is odd, then the size of $R_i$ is $k$ or $k+1$;  (ii) if $i$ is even, then the size of $R_i$ is $s$; (iii) $t$ is odd.
Hence we obtain the desired decomposition $V = A \cup B$ if we include the nodes in $R_1, R_3, \ldots$ to $B$ and include the nodes in $R_2, R_4, \ldots$ to $A$.
We show that such a decomposition of $S'$ into subpaths $R_1, R_2, \ldots, R_t$ exists.
Denote $z$ as the size of $S'$. We write $z = \alpha (s+k+1) + \beta$, where $\alpha > 0$ and $0 \leq \beta < s+k+1$ are integers.
Note that we must have $\alpha \geq k$ and $\beta \leq 2k$.
\begin{itemize}
\item For the case $\beta \geq k$, there is a decomposition  $R_1, R_2, \ldots, R_t$ satisfies the following conditions: (i) if $i$ is odd, then the size of $R_i$ is $k+1$ when $i <t$ or $\beta$ when $i=t$;  (ii) if $i$ is even, then the size of $R_i$ is $s$; (iii) $t = 2\alpha + 1$ is odd.
\item For the case $\beta < k$, there is a decomposition  $R_1, R_2, \ldots, R_t$ satisfies the following conditions: (i) if $i$ is odd, then the size of $R_i$ is $k$ when $i \in \{1, 3, 5, \ldots, 2(k-\beta)-1\} \cup \{t\}$ or $k+1$ otherwise;  (ii) if $i$ is even, then the size of $R_i$ is $s$; (iii) $t = 2\alpha + 1$ is odd.
\end{itemize}

For the rest of the proof, we show that in $O(\log^\ast n)$ time we can find the required independent set $I$. We prove the lemma by an induction on $L$. The base case of $L = 1$ is identical to the MIS problem.
Now consider $L > 1$. By induction hypothesis, we find an independent set $I'$ in $O(\log^\ast n)$ time such that each connected component induced by $V \setminus I'$ has at least $L'$ nodes and at most $2L'$ nodes, where $L' = \lfloor L/2 \rfloor$. Let $G'$ be the graph resulting from contracting all nodes in $V \setminus I'$, and we compute an MIS $I''$ on this graph $G'$, which can be done in $O(\log^\ast n)$ rounds in the original graph $G$. Note that each connected component $S$ of  $V \setminus I''$ has size at least $2L'+1 \geq L$ and at most $3(2L')+2$.
If the size of $S$ is higher than $2L$, then we can add some nodes in $S$ to the independent set $I''$ so that the component size of the remaining nodes in $S$ is within $[L, 2L]$.
  \end{proof}

\begin{lemma}
	If a feasible function $f$ exists, then there is an $O(\log^\ast n)$-round deterministic $\LOCAL$ algorithm for $\mathcal{P}$ on cycles.
\end{lemma}
\begin{proof}
	Given that the number of nodes $n$ is at least some large enough constant, in $O(\log^\ast n)$ rounds we can compute a decomposition $V = A \cup B$ such that each connected component of $A$ has size $2r$, and each connected component of $B$ has size within $[2\Lpump, 2\Lpump+1]$.
	This can be done using Lemma~\ref{lem:ruling-set} with $s = 2r$ and $k = 2\Lpump$.
	We further decompose each connected component $P$ of $B$ into two paths $P = P_1 \circ P_2$ in such a way that the size of both $P_1$ and $P_2$ are within the range $[\Lpump, \Lpump+1]$. We write $\Pset$ to denote the set of all these paths.
	
	Let $S$ be a connected component of $A$, and let $w_1$ and $w_2$ be its two neighboring paths in $\Pset$ so that $(w_1 \circ S \circ w_2)$ is a subpath of the underlying graph $G$. The output labels of $S$ are assigned either by labeling $S$ with $f(w_1 \circ S \circ w_2)$ or by labeling $S^R$ with $f(w_2^R \circ S^R \circ w_1^R)$.
	At this moment, all components of $A$ have been assigned output labels using $f$.
	By the feasibility of $f$, each connected component of $B$ is able to label itself output labels in such a way that the labeling of all nodes are locally consistent.
\end{proof}

\begin{lemma}
	If  there is an $o(n)$-round deterministic $\LOCAL$ algorithm $\mathcal{A}$ for $\mathcal{P}$ on cycles, then a feasible function $f$ exists.
\end{lemma}
\begin{proof}
	Fix $s$ to be some sufficiently large number, and fix $n = 8(s+\Lpump) + 2(2r)$.
We select $s$ to be large enough so that the runtime of $\mathcal{A}$ is smaller than $0.1s$.
	For any given directed path $w$ with $|w| \in [\Lpump, \Lpump+1]$, we fix $w^+$ as the result of applying the pumping lemma (Lemma~\ref{thm:pump1}) on $w$ so that the following two conditions are met: (i) $|w^+| \in [s, s+\Lpump]$ and (ii) $\type(w)=\type(w^+)$.

	\paragraph{\boldmath Constructing a Feasible Function $f$ by Simulating $\mathcal{A}$.}
	The function $f(w_1 \circ S \circ w_2)$ is constructed by simulating a given $o(n)$-round deterministic $\LOCAL$ algorithm for $\mathcal{P}$. The output labeling given by $f(w_1 \circ S \circ w_2)$ is exactly the result of simulating $\mathcal{A}$ on the path $P = w_1^+ \circ S \circ w_2^+$ while assuming the number of nodes of the underlying graph is $n$. Remember that the round complexity of $\mathcal{A}$ is $o(n)$ on $n$-node graphs. By setting $s$ to be large enough, the runtime of $\mathcal{A}$ can be made smaller than $0.1s$. Thus, the calculation of $f(w_1 \circ S \circ w_2)$ only depends on the IDs and the input labels of (i) the last $0.1 s$ nodes in $w_1^+$, (ii) all nodes in $S$, and (iii) the first $0.1 s$ nodes in $w_2^+$. In the calculation of $f(w_1 \circ S \circ w_2)$, the IDs of the nodes that participate in the simulation of $\mathcal{A}$ are chosen arbitrarily so long as they are distinct.

	\paragraph{\boldmath Feasibility of  $f$.}
	Now we verify that the function $f$ constructed above is  feasible.
	Consider any choices of paths $S_1 , S_2$ and $w_a, w_b, w_c, w_d$ such that $|S_1| = |S_2| = 2r$ and $\{|w_a|, |w_b|, |w_c|, |w_d|\} \subseteq [\Lpump, \Lpump+1]$.
	Define $P = w_a \circ S_1 \circ w_b \circ w_c \circ S_2 \circ w_d$, and let $G$ be the cycle graph formed by  connecting the two ends of the path $P$.
	To show that $f$ is feasible, we need to consider the
	following four ways of assigning output labels to $S_1 \cup S_2$.
	\begin{enumerate}[noitemsep]
		\item Label $S_1$ by $f(w_a \circ S_1 \circ w_b)$; label $S_2$ by $f(w_c \circ S_2 \circ w_d)$.
		\item Label $S_1$ by $f(w_a \circ S_1 \circ w_b)$; label $S_2^R$ by $f(w_d^R \circ S_2^R \circ w_c^R)$.
		\item Label $S_1^R$ by $f(w_b^R \circ S_1^R \circ w_a^R)$; label $S_2$ by $f(w_c \circ S_2 \circ w_d)$.
		\item Label $S_1^R$ by $f(w_b^R \circ S_1^R \circ w_a^R)$; label $S_2^R$ by $f(w_d^R \circ S_2^R \circ w_c^R)$.
	\end{enumerate}
	
	For each of the above four partial labelings of $P$, we need to show that the middle part $w_b \circ w_c$ can still be assigned output labels in such a way that the labeling of (i) the last $r$ nodes of $S_1$, (ii) all nodes in $w_b \circ w_c$, and (iii) the first $r$ nodes of $S_2$ are locally consistent.

	\paragraph{Proof of the First Case.}
	In what follows, we focus on the first case, i.e., the partial labeling is given by labeling $S_1$ by $f(w_a \circ S_1 \circ w_b)$ and labeling $S_2$ by $f(w_c \circ S_2 \circ w_d)$; the proof for the other three cases are analogous. In this case, we define $P' = w_a^+ \circ S_1 \circ w_b^+ \circ w_c^+ \circ S_2 \circ w_d^+$, and let $G'$ be the cycle graph  formed by connecting the two ends of $P'$.
	Note that the number of nodes in $G'$ is at most $8(s+\Lpump) + 2(2r) = n$.
	All we need to do is to find an output labeling $\LL$ of $G$ such that the following conditions are satisfied.
	\begin{itemize}
		\item[(a)] The output labels of $S_1$ is given  by $f(w_a \circ S_1 \circ w_b)$.
		\item[(b)] The output labels of $S_2$ is given  by $f(w_c \circ S_2 \circ w_d)$.
		\item[(c)] The labeling of (i) the last $r$ nodes of $S_1$, (ii) all nodes in $w_b \circ w_c$, and (iii) the first $r$ nodes of $S_2$ are locally consistent.
	\end{itemize}
	
	We first generate an output labeling $\LL'$ of $G'$ by  executing  $\mathcal{A}$ on $G'$ under the following ID assignment.
	The IDs of (i) the last $0.1 s$ nodes in $w_a^+$, (ii) all nodes in $S_1$, and (iii) the first $0.1 s$ nodes in $w_b^+$ are chosen as the ones used in the definition of  $f(w_a \circ S_1 \circ w_b)$. Similarly, the IDs of (i) the last $0.1 s$ nodes in $w_c^+$, (ii) all nodes in $S_2$, and (iii) the first $0.1 s$ nodes in $w_d^+$ are chosen as the ones used in the definition of  $f(w_c \circ S_2 \circ w_d)$. The IDs of the rest of the nodes are chosen arbitrarily so long as when we run  $\mathcal{A}$ on $G'$, no node sees two nodes with the same ID.
	Due to the way we define $f$, the output labeling $\LL'$ of the subpath
	$S_1$ is exactly given by $f(w_a \circ S_1 \circ w_b)$, and the output labeling $\LL'$ of $S_2$ is exactly  $f(w_c \circ S_2 \circ w_d)$. Due to the correctness of $\mathcal{A}$, $\LL'$ is a legal labeling.
	
	We transform the output labeling $\LL'$ of $G'$ to a desired output labeling $\LL$ of $G$.
	Remember that $G$ is the result of replacing the four subpaths $w^+$ of $G'$ by $w$, and we have $\type(w^+) = \type(w)$.
	In view of Lemma~\ref{lem:replace-path}, there is a legal labeling $\LL$ of $G$ such that all nodes in $S_1$ and $S_2$  are labeled the same as  in $G'$. Therefore, the labeling $\LL$ satisfies the above three conditions (a), (b), and (c).

	\paragraph{The Other Cases.}
	We briefly discuss how we modify the proof to deal with the other three cases. For example, consider the second case, where  the partial labeling is given by labeling $S_1$ by $f(w_a \circ S_1 \circ w_b)$ and labeling $S_2^R$ by $f(w_d^R \circ S_2^R \circ w_c^R)$. In this case, the path $P'$ is defined as  \[P' = w_a^+ \circ S_1 \circ w_b^+ \circ \left((w_c^R)^+\right)^R \circ S_2^R \circ \left((w_d^R)^+\right)^R.\] During the ID assignment of $G'$, the IDs of (i) the last $0.1 s$ nodes in $w_c^+$, (ii) all nodes in $S_2$, and (iii) the first $0.1 s$ nodes in $w_d^+$ are now chosen as the ones used in the definition of  $f(w_d^R \circ S_2^R \circ w_c^R)$. Using such an ID assignment, the output labeling $\LL'$ of $S_2^R$ as the result of executing $\mathcal{A}$  on $G'$ will be exactly the same as the output labeling given by  $f(w_d^R \circ S_2^R \circ w_c^R)$.
	The rest of the proof is the same.
\end{proof}

Theorem~\ref{thm-cycle-gap-2} follows from the above two lemmas. The decidability result is due to the simple observation that whether a feasible function exists is decidable.

\subsection{Partitioning a Cycle\label{sec.orient2}}

In the following sections, we prove the decidability result associated with the $\omega(1)$---$o(\log^\ast n)$ gap. In this proof, we also define a feasible function, prove its decidability, and show the existence given an $o(\log^\ast n)$-time algorithm. The main challenge here is that an MIS cannot be computed in $O(1)$ time. To solve this issue, we decompose a cycle into paths with unrepetitive patterns and paths with repetitive patterns. For paths with unrepetitive patterns, we are able to compute a sufficiently well-spaced MIS in $O(1)$ time by making use of the irregularity of the input patterns.

Section~\ref{sec.orient2} considers an $O(1)$-round algorithm that partitions a cycle into some short paths and some paths that have a repeated input pattern.
Section~\ref{sec.func} defines a {\em feasible function} whose existence characterizes the $O(1)$-round solvable \lcl\ problems.
In Section~\ref{sec.cycle-main}, we prove Theorem~\ref{thm-cycle-gap-1}.

\paragraph{Partitioning an Undirected Cycle into Directed Paths.}
Let $G$ be a cycle graph.
An orientation of a node $v$ is an assignment to one of its neighbor, this can be specified using port-numbering.
 An orientation of the nodes in $G$ is called {\em $\ell$-orientation} if the following condition is met. If $|V(G)| \leq \ell$, then all nodes in $G$ are oriented to the same direction. If $|V(G)| > \ell$, then each node $v \in V(G)$ belongs to a path $P$ such that (i) all nodes in $P$ are oriented to the same direction, and (ii) the number of nodes in $P$ is at least $\ell$.
In $O(1)$ rounds we can compute an $\ell$-orientation of $G$ for any constant $\ell$.

\begin{lemma}[\cite{chang17hierarchy}]\label{lem:ori}
	Let $G$ be a cycle graph.
	Let $\ell$ be a constant.
	There is a deterministic $\LOCAL$ algorithm that computes an $\ell$-orientation of $G$ in $O(1)$ rounds.
\end{lemma}

In this section, we will use a generalization of an $\ell$-orientation that satisfies an
additional requirement that the input labels of each directed path $P$ in the decomposition with $|V(P)| > 2 \Lwidth$ (where $2 \Lwidth$ is a threshold) must form a periodic string (whose period length is at most $\Lpattern$).

A string $w \in \LabelIn^\ast$ is called {\em primitive} if $w$ cannot be written as $x^i$ for some $x \in \LabelIn^\ast$ and $i \geq 2$.
Let $G$ be a cycle graph or a path graph where each node $v \in V(G)$ has an input label from $\LabelIn$.
We define an {\em $(\Lwidth, \Lcount , \Lpattern)$-partition} as a partition of  $G$ into a set of connected subgraphs $\Pset$ meeting the following criteria. We assume $|V(G)| > 2 \Lwidth$ and $\Lpattern \geq \Lwidth$.

\begin{framed}
	\begin{description}[leftmargin=0pt,topsep=0pt]
		\item[Direction and Minimum Length:] For each $P \in \Pset$, the nodes in $P$ are oriented to the same direction, and $|V(P)|\geq \Lwidth$.
		\item[Short Paths:] Define $\Pshort$ as the subset of $\Pset$ that contains paths having at most $2 \Lwidth$ nodes. For each directed path $P=(v_1, \ldots, v_k) \in \Pshort$, each node $v_i$  in $P$ knows its rank $i$.
		\item[Long Paths:] Define $\Plong = \Pset \setminus \Pshort$. Then the input labeling of the nodes in $P$ is of the form $w^{k}$ for some primitive string $w \in \LabelIn^\ast$ such that $|w| \leq \Lpattern$ and $k \geq \Lcount$. Moreover, each node $v$ in $P$ knows the string $w$.
	\end{description}
\end{framed}

Note that $\Pset$ may contain a cycle. This is possible only when $G$ is a cycle where the input labeling is a repetition (at least $\Lcount$ times) of a primitive string $w \in \LabelIn^\ast$ of length at most $\Lpattern$. In this case, we must have $\Pset = \Plong = \{G\}$.
Otherwise, $\Pset$ contains only paths.

The goal of this section is to show that an $(\Lwidth, \Lcount, \Lpattern)$-partition can be found in $O(1)$ rounds.
First of all, in Lemma~\ref{lem:MIS} we demonstrate how we can break symmetry in $O(1)$ rounds given that the underlying graph is directed and the input labels does not form long periodic strings.
Let $G$ be a path or a cycle. A set $I \subseteq V(G)$ is called an {\em $(\alpha,\beta)$-independent set} if the following conditions
are met: (i) $I$ is an independent set, and $I$ does not contain either endpoint of $G$ (if $G$ is a path), and (ii) each connected component induced by $V \setminus I$ has at least $\alpha$ nodes and at most $\beta$ nodes, unless $|V| \leq \alpha$, in which
case we allow $I = \emptyset$.
Note that finding an $(\alpha,\beta)$-independent set takes $O(\log^\ast n)$ rounds in general, but in Lemma~\ref{lem:MIS} we show that by leveraging the ``irregularity'' of input labels, we can do this in $O(1)$ rounds on directed paths or cycles without periodic patterns.

\begin{lemma}\label{lem:MIS}
	Let $\gamma$ and $\ell$ be any two constants with $\ell \geq \gamma$.
	Let $G$ be a directed cycle or a directed path that does not contain any subpath of the form $w^{x}$, with $|w| \leq \gamma$ and $|w^{x}| \geq \ell$.
	There is a deterministic $\LOCAL$ algorithm that computes an $(\gamma, 2\gamma)$-independent set $I$ of $G$ in $O(1)$ rounds.
\end{lemma}
\begin{proof}
	For the case $G$ is a directed path $P=(s, \ldots, t)$, define $V'$ as the set of nodes in $G$ whose distance to $t$ is at least $\ell-1$.
	For the case $G$ is a directed cycle, define $V' = V(G)$.
	In what follows, we focus on finding an $(\gamma, 2\gamma)$-independent set $I'$ of the nodes in $V'$.
	Extending the set $I'$ to produce the desired independent set $I$ can be done with extra $O(1)$ rounds.
	
	Recall that $G$ is directed.
	Define the color of a node $v \in V'$ by the sequence of the $\ell$ input labels of $v$ and the $\ell -1$ nodes following $v$ in $G$. For each node $v \in V'$, there is no other node within distance $\gamma$ to $v$ having the same color as $v$, since otherwise we can find a subpath whose input labels form a string $w^{x}$, with $|w| \leq \gamma$ and $|w^{x}| \geq \ell$.
	By applying the standard procedure that computes an MIS from a coloring,
	within $O(1)$ rounds  a $(\gamma, 2\gamma)$-independent set $I'$ can be obtained.
\end{proof}

Using Lemma~\ref{lem:MIS}, we  first show that  an $(\Lwidth, \Lcount , \Lpattern)$-partition can be found in $O(1)$ rounds for the case $G$ is {\em directed}. That is, all nodes in $G$ are initially oriented to the same direction, and we are allowed to re-orient the nodes.

\begin{lemma}\label{lem:MISorient-base}
	Let $G$ be a directed cycle or a directed path where each node $v \in V(G)$ has an input label from $\LabelIn$, and $|V(G)| > 2 \Lwidth$.
	Let $\Lwidth, \Lcount , \Lpattern$ be three constants such that $\Lpattern \geq \Lwidth$.
	There is a deterministic $\LOCAL$ algorithm that computes an $(\Lwidth, \Lcount , \Lpattern)$-partition in $O(1)$ rounds
\end{lemma}
\begin{proof}
	Let $(w_1, w_2, \ldots, w_k)$ be {\em any} ordering of the primitive strings in $\LabelIn^\ast$ of length at most $\Lpattern$.
	First, construct a set of subgraphs $\Plong$ as follows.
	Initialize $U = V(G)$ and $\Plong = \emptyset$. For $i=1$ to $k$, execute the following procedure.
	Let $S_i$ be the set of maximal-size connected subgraphs formed by nodes in $U$ such that the input labels form the string
	$w_i^x$ with $x \geq \Lcount + 2 \Lwidth$.
	Each node $v \in U$ in $O(1)$ rounds checks if $v$ belongs to a subgraph in $S_i$; if so, remove $v$ from $U$.
	For each $P \in S_i$, define $P'$ as follows.
	If $P$ is a cycle, then $P' = P$.
	If $P$ is a path, then $P'$ is the result of removing all nodes that are within distance $\Lwidth |w_i| - 1$ to an endpoint in $P$.
	Note that each node $v$ in $P$ knows whether $v$ belongs to $P'$.
	Define $S_i' = \{P' | P \in S_i\}$, and then update $\Plong \leftarrow \Plong \cup S_i'$.
	
	It is straightforward to verify that each path or cycle $P \in \Plong$ satisfies the requirement in the definition of $(\Lwidth, \Lcount , \Lpattern)$-partition.
	Define the set of subgraphs $\Pset_{\text{irreg}}$ as the connected components of the nodes not in any subgraph in $\Plong$.
	Define $\ell = (\Lpattern +  2 \Lwidth) \cdot \Lcount$.
	By our construction, the input labeling in each subgraph $P \in \Pset_{\text{irreg}}$ does not contain any substring $w^{x}$, with $1 \leq |w| \leq \Lpattern$ and $|w^{x}| \geq \ell$. An $(\Lpattern, 2\Lpattern)$-independent set of each $P \in \Pset_{\text{irreg}}$ can be computed using Lemma~\ref{lem:MIS} in $O(1)$ rounds. Observe that each subgraph $P \in \Pset_{\text{irreg}}$ has at least $\Lwidth$ nodes. Given an $(\Lpattern, 2\Lpattern)$-independent set of a subgraph $P \in \Pset_{\text{irreg}}$, in $O(1)$ rounds $P$ can be partitioned into subpaths, each of which contains at least $\Lpattern$ nodes and at most $2 \Lpattern$ nodes. This finishes the construction of an $(\Lwidth, \Lcount, \Lpattern)$-partition.
\end{proof}

Combining Lemma~\ref{lem:MISorient-base} and Lemma~\ref{lem:ori}, we are able to construct an $(\Lwidth, \Lcount, \Lpattern)$-partition  in $O(1)$ rounds for undirected graphs.

\begin{lemma}\label{lem:MISorient}
	Let $G$ be a  cycle or a path where each node $v \in V(G)$ has an input label from $\LabelIn$, and $|V(G)| > 2 \Lwidth$.
	Let $\Lwidth, \Lcount, \Lpattern$ be three constants such that $\Lpattern \geq \Lwidth$.
	There is a deterministic $\LOCAL$ algorithm that computes an $(\Lwidth, \Lcount, \Lpattern)$-partition in $O(1)$ rounds
\end{lemma}
\begin{proof}
	The algorithm is as follows.
	Compute an $\ell$-orientation of $G$ by Lemma~\ref{lem:ori} in $O(1)$ rounds with $\ell = 2 \Lwidth + 1$.
	For each maximal-length connected subgraph $P$ where each constituent node is oriented to the same direction,
	find an $(\Lwidth, \Lcount , \Lpattern)$-partition of $P$ in $O(1)$ rounds by Lemma~\ref{lem:MISorient-base}.
\end{proof}

\subsection{Feasible Function\label{sec.func}}
The goal of this section is to define a {\em feasible function} whose existence characterizes the $O(1)$-round solvable \lcl\ problems.
With respect to an \lcl\ problem $\mathcal{P}$ and a function $f$ which takes a string $w \in \LabelIn^{k}$ with $1 \leq k \leq \Lpump$ as input, and returns a string $f(w) \in \LabelOut^{k}$, we define some partially  or completely labeled path graphs which are used in the definition of a feasible function.
\begin{description}
	\item[\boldmath Completely Labeled Graph $\GG_{w,z}$:]
	Let $w \in \LabelIn^{\ast}$ be any string of length at least 1 and at most $\Lpump$.
	Let $z$ be any non-negative integer.
	Define $\GG_{w,z}=(G_{w,z}, \LL)$ as follows.
	The graph $G_{w,z}$ is a path of the form $w^{r} \circ w^z \circ w^{r}$.
	The labeling $\LL$ is a complete labeling of the form ${f(w)}^{z+2r}$.
	Define $\midd(G_{w,z})$ as the middle subpath $w^z$ of $G_{w,z}$.
	\item[\boldmath Partially Labeled Graph $\GG_{w_1,w_2,S}$:]
	Let $w_1, w_2 \in \LabelIn^{\ast}$ be any two strings of length at least 1 and at most $\Lpump$.
	Let $S \in \LabelIn^{\ast}$ be any string (can be empty).
	Define $\GG_{w_1,w_2,S}=(G_{w_1,w_2,S}, \LL)$ as follows.
	The graph $G_{w_1,w_2,S}$ is the path of the form $w_1^{\Lpump+2r} \circ S \circ w_2^{\Lpump+2r}$.
	The labeling $\LL$ is a partial labeling of $G_{w_1,w_2,S}$ which fixes the output labels of the first $2r |w_1|$ and the last $2r |w_2|$ nodes by $f(w_1)^{2r}$ and $f(w_2)^{2r}$, respectively.
	Define $\midd(G_{w_1,w_2,S})$ as the middle subpath $w_1^{\Lpump+r} \circ S \circ w_2^{\Lpump+r}$ of $G_{w_1,w_2,S}$.
	\item[Feasible Function:] We call $f$ a {\em feasible function} if the following conditions are met: (i) For each $\GG_{w,z}=(G_{w,z}, \LL)$, the complete labeling $\LL$  is locally consistent at all nodes in $\midd(G_{w,z})$. (ii) Each partially labeled graph $\GG_{w_1,w_2,S}$ admits a complete labeling $\LL_\diamond$ that is locally consistent at all nodes in  $\midd(\GG_{w_1,w_2,S})$.
\end{description}

\begin{lemma}\label{lem:func-decide-cycle}
	Given an \lcl\ problem $\mathcal{P}$ on cycle graphs. It is decidable whether there is a feasible function.
\end{lemma}
\begin{proof}
	Note that it is not immediate from its definition as to whether a feasible function exists is decidable, since there appears to be {\em infinitely} many graphs $\GG_{w,z}$ and $\GG_{w_1,w_2,S}$ needed to be examined.
	However, the following simple observations show that it suffices to check only a constant number of these graphs.
	\begin{itemize}
		\item If the complete labeling $\LL$ of $\GG_{w,1}=(G_{w,1}, \LL)$ is locally consistent at all nodes in $\midd(G_{w,1})$, then for all $z \geq 1$, the complete labeling $\LL$ of $\GG_{w,z}=(G_{w,z}, \LL)$ is also locally consistent at all nodes in $\midd(G_{w,z})$.
		\item If $\GG_{w_1,w_2,S}$ admits a complete labeling $\LL_\diamond$ that is locally consistent at all nodes in  $\midd(\GG_{w_1,w_2,S})$, then  for each $S'$ such that $\type(S) = \type(S')$, the partially labeled graph $\GG_{w_1,w_2,S'}$ also admits a complete labeling $\LL_\diamond$ that is locally consistent at all nodes in  $\midd(\GG_{w_1,w_2,S'})$. This is due to Lemma~\ref{lem:replace-path}.
	\end{itemize}
	
	Therefore, to decide whether a function $f$ is feasible, we only need to check all possible $\GG_{w,z}$ and $\GG_{w_1,w_2,S}$. For each $w$ we only need to consider the graph $\GG_{w,z}$ with $z = 1$.
	For each $w_1$ and $w_2$, we do not need to go over all $S$; we only need to consider (i) the empty string $S = \emptyset$, and (ii) for each type $\tau$, a string $S \in \LabelIn^\ast$ such that $\type(S)=\tau$.
	By Lemma~\ref{thm:pump1}, for each type $\tau$, there exists $P \in \LabelIn^x$ with $x \leq \Lpump$ such that $\type(P)=\tau$.
	Therefore, a string $S$ with $\type(S)=\tau$ can be found in bounded amount of time; also note that the number of types is bounded; see Lemma~\ref{lem-type-number}.
\end{proof}

For the rest of this section, we show that as long as the deterministic $\LOCAL$ complexity of $\mathcal{P}$ is $o(\log^\ast n)$ on cycle graphs, there exists a feasible function $f$. In Lemma~\ref{lem:func-aux} we show how to extract a  function $f$ from a given $o(\log^\ast n)$-round deterministic $\LOCAL$ algorithm $\mathcal{A}$, and then in Lemma~\ref{lem:func} we prove that such a function $f$ is feasible.
Intuitively, Lemma~\ref{lem:func-aux} shows that there exists an ID-assignment  such that when we run  $\mathcal{A}$ on a subpath whose input labeling is a repetition of a length-$k$ pattern $w$, the output labeling is also a repetition of a length-$k$ pattern $w'$. The function $f$ will be defined as  $f(w) = w'$.

\begin{lemma}\label{lem:func-aux}
	Let $\mathcal{A}$ be any  deterministic $\LOCAL$ algorithm that solves $\mathcal{P}$ in $t(n) = o(\log^\ast n)$ rounds.
	Then there is a number $n'$ and function $f$ which takes a string $w \in \LabelIn^{k}$ with $1 \leq k \leq \Lpump$ as input, and returns a string $f(w) \in \LabelOut^{k}$ meeting the following condition. For any  $P = w^{i} \circ w^{2r+1} \circ w^i$ such that $|w^i| \geq t(n')$ and $1 \leq |w| \leq \Lpump$, there is an assignment of distinct $\Theta(\log n')$-bit IDs to the nodes in $P$ such that the following is true. Simulating $\mathcal{A}$ on $P$ while assuming that the total number of nodes in the underlying graph is $n'$ yields the output labeling $f(w)^{2r+1}$ for the middle subpath $w^{2r+1}$.
\end{lemma}
\begin{proof}
	In this proof we assume that there is no such a number $n'$.
	Then we claim that using $\mathcal{A}$ it is possible to obtain a deterministic $\LOCAL$ algorithm for MIS on an $n$-node directed cycle $G$ {\em without input labeling}, in $O(t(n)) + O(1) = o(\log^\ast n)$ rounds.
	This contradicts the well-known $\Omega(\log^\ast n)$ lower bound for MIS~\cite{Linial1992}.
	
	Let $G$ be an $n$-node directed cycle without input labeling.
	The MIS algorithm on $G$ is described as follows.
	Let $w \in \LabelIn^{k}$ with $1 \leq k \leq \Lpump$  be chosen such that for any function $f$, the string $f(w) \in \LabelOut^{k}$ does not satisfy the conditions stated in the lemma for the number $n' = nk$.
	Define $G'$ as the graph resulting from replacing each node $v \in V(G)$ with a path $w$.
	We can simulate the imaginary graph $G'$ in the communication network $G$ by letting each node $v \in V(G)$ simulate a path $w$.
	
	We execute the algorithm $\mathcal{A}$ on $G'$ while assuming that the total number of nodes is $n'$.
	The execution takes $t(n') = O(t(n))$ rounds.
	For each node $v \in V(G)$, define the color of $v$ as the  sequence of the output labels of the path $w^{2r}$ simulated by the node  $v$ and the $2r-1$ nodes following $v$ in the directed cycle $G$. This gives us a proper $O(1)$-coloring, since otherwise there must exist a subpath $P = w^{2r+1}$ of $G'$ such that the output labeling of $P$ is of the form $y^{2r+1}$ for some $y$, contradicting our choice of $w$. Using the standard procedure of computing an MIS from a coloring, with extra $O(1)$ rounds, an MIS of $G$ can be obtained.
	
	Note that there is a subtle issue about how we set the IDs of nodes in $V(G')$. The following method is guaranteed to output distinct IDs. Let $v \in V(G)$, and let $u_1, \ldots, u_k$ be the nodes in $V(G')$ simulated by $v$. Then we may use $\ID(u_i) = k \cdot \ID(v) + i$.
\end{proof}

\begin{lemma}\label{lem:func}
	Suppose that the deterministic $\LOCAL$ complexity of $\mathcal{P}$ is $o(\log^\ast n)$ on cycle graphs. Then there exists a feasible function $f$.
\end{lemma}
\begin{proof}
	Let $\mathcal{A}$ be any  deterministic $\LOCAL$ algorithm that solves $\mathcal{P}$ in $t(n) = o(\log^\ast n)$ rounds.
	Let $n'$ and $f$ be chosen to meet the conditions in Lemma~\ref{lem:func-aux} for $\mathcal{A}$.
	The goal of the proof is to show that $f$ is a feasible function.
	According to the conditions specified in Lemma~\ref{lem:func-aux} for the function $f$,
	we already know that the complete labeling $\LL$ of each $\GG_{w,z}=(G_{w,z}, \LL)$ is locally consistent at all nodes in $\midd(G_{w,z})$.
	Therefore, all we need to do is the following. For each partially labeled graph $\GG_{w_1,w_2,S}$, find a complete labeling $\LL_\diamond$ that is locally consistent at all nodes in $\midd(\GG_{w_1,w_2,S})$.
	
	Given the three parameters $w_1$, $w_2$, and $S$, define $G$ as the cycle resulting from linking the two ends of the path
	$w_1^{\Lpump} \circ w_1^{2r+1} \circ w_1^{\Lpump} \circ S \circ w_2^{\Lpump} \circ w_2^{2r+1} \circ w_2^{\Lpump}$.
	Define $\LL$ as the partial labeling of $G$ which fixes the output labeling of the two subpaths $w_1^{2r+1}$ and $w_2^{2r+1}$ by
	$f(w_1)^{2r+1}$ and $f(w_2)^{2r+1}$, respectively.
	We write $P_1^{\text{mid}}$ and $P_2^{\text{mid}}$ to denote the two subpaths $w_1^{2r+1}$ and $w_2^{2r+1}$, respectively.
	
	In what follows, we show that the partially labeled graph $\GG=(G, \LL)$ admits a legal labeling $\LL_\diamond$.
	Since $\GG_{w_1,w_2,S}$ is a subgraph of $\GG = (G,\LL)$,  such a legal labeling $\LL_\diamond$ is also a complete labeling  of  $\GG_{w_1,w_2,S}$ that is locally consistent at all nodes in $\midd(\GG_{w_1,w_2,S})$.
	
	For the rest of the proof, we show the existence of $\LL_\diamond$. This will be established by applying a pumping lemma.
	Define the graph $G'$ as the result of the following operations on $G$.
	\begin{itemize}
		\item Replace the two subpaths $w_1^{\Lpump}$ by $w_1^{x}$, where the number $x$ is chosen such that  $x|w_1| \geq 2t(n') + r$, and $\type(w_1^{\Lpump}) = \type(w_1^x)$.
		\item Replace the two subpaths $w_2^{\Lpump}$ by $w_2^{y}$, where the number $y$ is chosen such that  $y|w_2| \geq 2t(n') + r$, and $\type(w_2^{\Lpump}) = \type(w_2^y)$.
	\end{itemize}
	The existence of the numbers $x$ and $y$ above is guaranteed by Lemma~\ref{thm:pump2}.
	The IDs of nodes in $G'$ are assigned as follows.
	For $i = 1,2$, select the IDs of the nodes in $\bigcup_{v \in P_i^{\text{mid}}} N^{t(n')}(v)$ in such a way that the output labeling of $P_i^{\text{mid}}$ resulting from executing $\mathcal{A}$ on $G'$ while assuming that the total number of nodes is $n'$ is ${f(w_i)}^{2r+1}$. The existence of such an ID assignment is guaranteed by Lemma~\ref{lem:func-aux}.
	For all remaining nodes in $G'$, select their IDs in such a way that all nodes in $N^{r + t(n')}(v)$ receive distinct IDs, for each $v \in V(G')$. This ensures that the outcome of executing $\mathcal{A}$ on $G'$ while assuming that the total number of nodes is $n'$ is a legal labeling.
	
	Let $\LL_\diamond'$ be the legal labeling of $G'$ resulting from executing $\mathcal{A}$ with the above IDs while pretending that the total number of nodes is $n'$. Note that $\LL_\diamond'$ must label $P_1^{\text{mid}}$ and $P_2^{\text{mid}}$ by ${f(w_1)}^{2r+1}$ and ${f(w_2)}^{2r+1}$, respectively.
	A desired legal labeling $\LL_\diamond$ of $\GG$ can be obtained from the legal labeling $\LL_\diamond'$ of $G'$ by applying Lemma~\ref{lem:replace-path}, as we have $\type(w_1^{\Lpump}) = \type(w_1^x)$ and $\type(w_2^{\Lpump}) = \type(w_2^y)$.
\end{proof}

\subsection{\boldmath The \texorpdfstring{$\omega(1)$}{omega(1)}---\texorpdfstring{$o(\log^\ast n)$}{o(log* n)} Gap \label{sec.cycle-main}}
In this section we prove that it is decidable whether a given \lcl\ problem $\mathcal{P}$ has complexity  $\Omega(\log^\ast n)$ or  $O(1)$ on cycle graphs.

\begin{lemma}\label{lem:alternate-func}
	Let $f$ be any feasible function.
	Let $G$ be any cycle graph.
	Let $\Pset$ be any set of disjoint subgraphs in $G$ such that the input labeling of each $P \in \Pset$ is of the form $w^{x}$ such that $x \geq 2\Lpump + 2r$, and $w \in \LabelIn^{k}$  is a string with  $1 \leq k \leq \Lpump$.
	For each $P \in \Pset$, define the subgraph $P'$ as follows.
	If $P$ is a cycle, define $P' = P$.
	If $P$ is a path, write $P = w^{\Lpump} \circ w^{i} \circ w^{\Lpump}$, and define $P'$ as the middle subpath $w^i$.
	Let $\LL$ be a partial labeling of $G$ defined as follows.
	For each $P = w^x \in \Pset$, fix the output labels of each subpath $w$ of $P'$ by $f(w)$.
	Then $\GG=(G, \LL)$ admits a legal labeling $\LL_\diamond$.
\end{lemma}
\begin{proof}
	Define $V_1$  as the set of all nodes such that $v \in V_1$ if $v$ belongs to the middle subpath $w^{j}$ of some path $P = w^{\Lpump} \circ w^{r} \circ w^{j}\circ w^{r} \circ w^{\Lpump} \in \Pset$. By the definition of feasible function, $\LL$ is already locally consistent at all nodes in $V_1$.
	Thus, all we need to do is to construct a complete labeling $\LL_\diamond$ of $\GG=(G, \LL)$, and argue that $\LL_\diamond$ is locally consistent at all nodes in $V_2 = V(G) \setminus V_1$.
	
	There are two easy special cases. If $\Pset = \emptyset$, then no output label of any node in $G$ is fixed, and so $\GG$ trivially admits a legal labeling. If $\Pset$ contains a cycle, then $\Pset = \{G\}$, and hence $\LL$ is already a legal labeling as $V_1 = V(G)$.
	
	In subsequent discussion, we restrict ourselves to the case that $\Pset$ is non-empty and contains only paths.
	The output labeling $\LL_\diamond$ is constructed as follows.
	Define $\Pset_{\text{unlabeled}}$ as the maximal-length subpaths of  $G$ that are not assigned any output labels by $\LL$.
	A path $P \in \Pset_{\text{unlabeled}}$  must be of the form $w_1^{\Lpump} \circ S \circ w_2^{\Lpump}$, where $w_1, w_2 \in \LabelIn^{\ast}$ are two strings of length at least 1 and at most $\Lpump$, and $S \in \LabelIn^{\ast}$ can be any string (including the empty string). Given $P \in \Pset_{\text{unlabeled}}$, we make the following definitions.
	\begin{itemize}
		\item Define $P^{+}$ as the subpath of $G$ that includes $P$ and the $r |w_1|$ nodes preceding $P$, and the $r |w_2|$ nodes following $P$ in the graph $G$.
		Note that the set $V_2$ is exactly the union of nodes in $P^+$ for all $P \in \Pset_{\text{unlabeled}}$.
		\item Define $P^{++}$ as the subpath of $G$ that includes $P$ and the $2r |w_1|$ nodes preceding $P$, and the $2r |w_2|$ nodes following $P$ in the graph $G$.
		The path $P^{++}$ must be of the form $w_1^{\Lpump+2r} \circ S \circ w_2^{\Lpump+2r}$, and the labeling $\LL$ already fixes the output labels of the first $2r |w_1|$ and the last $2r |w_2|$ nodes of $P^{++}$  by $f(w_1)^{2r}$ and $f(w_2)^{2r}$, respectively.
	\end{itemize}
	Observe that the path $P^{++} = w_1^{\Lpump+2r} \circ S \circ w_2^{\Lpump+2r}$ together with the labeling $\LL$ is exactly the partially labeled graph $\GG_{w_1,w_2,S}$.
	We assign the output labels to the nodes in $P$ by the labeling $\LL_\diamond$ guaranteed in the definition of feasible function.
	It is ensured that the labeling of all nodes within $P^+$ are locally consistent.
	By doing so for each $P \in \Pset_{\text{unlabeled}}$, we obtain a desired complete labeling that is locally consistent at all nodes in $V_2$.
\end{proof}

\begin{lemma}\label{lem:cycle-final-lem}
	Suppose that there is a feasible function $f$ for the \lcl\ problem $\mathcal{P}$.
	Then there is an $O(1)$-round deterministic $\LOCAL$ algorithm $\mathcal{A}$ on cycle graphs.
\end{lemma}
\begin{proof}
	The first step of the algorithm $\mathcal{A}$  is to compute an $(\Lwidth, \Lcount , \Lpattern)$-partition in $O(1)$ rounds by Lemma~\ref{lem:MISorient}.
	We set $\Lcount = 2 \Lpump + 2r$ and $\Lwidth = \Lpattern = \Lpump$.
	We assume $|V(G)| > 2 \Lwidth$.
	Recall that an $(\Lwidth, \Lcount , \Lpattern)$-partition decomposes the cycle $G$ into two sets of disjoint subgraphs $\Pshort$ and $\Plong$.

	Define $G'$ as the graph resulting from applying the following operations on $G$. For each $P \in \Pshort$, replace the path $P$ by the path $P^{\ast} = x \circ y^i \circ z$ such that $i = \Lcount$, $1 \leq |y| \leq \Lpattern$, and the type of $P^\ast$ is the same as the type of $P$. The path $P^\ast$ is obtained via Lemma~\ref{thm:pump1}. Note that each path $P \in \Pshort$ has at least $\Lwidth = \Lpump$ nodes and at most $2\Lwidth = 2 \Lpump$ nodes. Define $\Pset^\ast$ as the set of all $P^{\ast}$ such that $P \in \Pshort$. The graph $G'$ is simulated in the communication graph $G$ by electing a leader for each path $P \in \Pshort$ to simulate $P^{\ast}$.
	
	Calculate a partial labeling $\LL'$ of $G'$ using  the feasible function $f$ as follows.
	Recall $\Lcount = 2 \Lpump + 2r$.
	For each $P^{\ast} = x \circ y^{\Lpump} \circ y^{2r} \circ y^{\Lpump} \circ z\in \Pset^\ast$, label the middle subpath $y^{2r}$ by the function $f$.
	For each $P=w^{\Lpump} \circ w^{i} \circ w^{\Lpump} \in \Plong$, label the middle subpath $w^{i}$ by $f(w)^i$. Even though a path $P \in \Plong$ can have $\omega(1)$ nodes, this step can be done locally in $O(1)$ rounds due to the following property of $(\Lwidth, \Lcount, \Lpattern)$-partition. All nodes in a path $P \in \Plong$ agree with the same direction and know the primitive string $w$.

	By Lemma~\ref{lem:alternate-func}, the remaining unlabeled nodes in $G'$ can be labeled to yield a legal labeling of $G'$.
	This can be done in $O(1)$ rounds since the connected components formed by unlabeled nodes have at most $O(1)$ nodes.
	Given any valid labeling of $G'$, a legal labeling of $G$ can be obtained by applying Lemma~\ref{lem:replace-path} in $O(1)$ rounds. Remember that $\type(P) = \type(P^\ast)$ for each $P \in \Pshort$, and $G'$ is exactly the result of replacing each $P \in \Pshort$ by $P^\ast$.
\end{proof}

\begin{figure}
	\centerline
	{\includegraphics[width=1\textwidth]{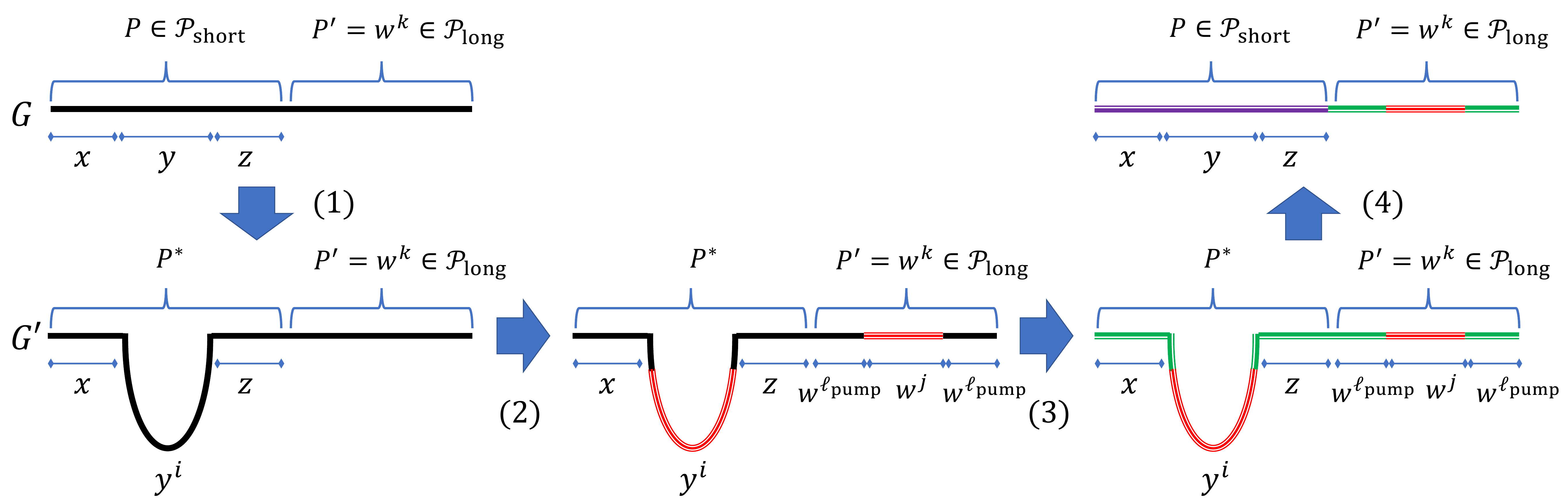}}
	\caption{\label{fig:upper-bound-proof-illustration} Illustration of Lemma~\ref{lem:cycle-final-lem}.}
\end{figure}

See Figure~\ref{fig:upper-bound-proof-illustration} for an illustration of Lemma~\ref{lem:cycle-final-lem}:
(1) applying a pumping lemma to extend each path  $P \in \Pshort$;
(2) labeling the middle subpath $y^{2r}$ of $P^{\ast} = x \circ y^{\Lpump} \circ y^{2r} \circ y^{\Lpump} \circ z\in \Pset^\ast$ and the  middle subpath $w^{j}$ of $P'=w^{\Lpump} \circ w^{i} \circ w^{\Lpump} \in \Plong$  by the function $f$;
(3) the remaining unlabeled nodes in $G'$ can be labeled to yield a legal labeling of $G'$ by Lemma~\ref{lem:alternate-func};
(4) since $\type(P) = \type(P^\ast)$ for each $P \in \Pshort$, we can recover a legal labeling of $G$ by re-labeling nodes in each $P \in \Pshort$.

Combining  Lemma~\ref{lem:func-decide-cycle}, Lemma~\ref{lem:func}, and Lemma~\ref{lem:cycle-final-lem}, we have proved Theorem~\ref{thm-cycle-gap-1}. That is, for any \lcl\ problem $\mathcal{P}$ on cycle graphs, its deterministic $\LOCAL$ complexity is either  $\Omega(\log^\ast n)$ or  $O(1)$. Moreover, there is an algorithm that decides whether $\mathcal{P}$ has complexity  $\Omega(\log^\ast n)$ or  $O(1)$ on cycle graphs; for the case the complexity is $O(1)$, the algorithm outputs a description of an $O(1)$-round deterministic $\LOCAL$ algorithm that solves $\mathcal{P}$.

\section*{Acknowledgments}
Many thanks to Laurent Feuilloley, Juho Hirvonen, Janne H. Korhonen, Christoph Lenzen, Yannic Maus, and Seth Pettie for discussions, and to anonymous reviewers for their helpful comments on previous versions of this work. This work was supported in part by the Academy of Finland, Grant 285721.

{
	\urlstyle{sf}
	\DeclareUrlCommand{\Doi}{\urlstyle{same}}
	\renewcommand{\doi}[1]{\href{http://dx.doi.org/#1}{\footnotesize\sf doi:\Doi{#1}}}
	\bibliographystyle{plainnat}
	\bibliography{references}
}

\end{document}